\documentclass[amsmath,amssymb,showpacs,superscriptaddress]{revtex4}

\usepackage{color}
\usepackage{graphicx}
\usepackage{dcolumn}
\usepackage{bm}
\usepackage{amsmath,amsthm,amssymb}
\usepackage{hyperref}

\newtheorem{thm}{Theorem}[section]

\newcommand{\be}{\begin{equation}}
\newcommand{\ee}{\end{equation}}

\newcommand{\fa}{\Psi}

\newcommand{\fad}{\Psi^\dagger}

\newcommand{\lam}{\lambda}
\newcommand{\ra}{\rangle}
\newcommand{\la}{\langle}
\newcommand{\inti}{\int_{-\infty}^{+\infty}}
\newcommand{\alf}{\alpha}
\newcommand{\bet}{\beta}
\newcommand{\nuf}{\nu}
\newcommand{\Imm}{\Im }
\newcommand{\Ree}{\Re}

\begin{document}

\title{One-Dimensional Impenetrable Anyons in Thermal Equilibrium. III.
Large distance asymptotics of the space correlations.}

\author{Ovidiu I. P\^{a}\c{t}u}

\affiliation{C.N. Yang Institute for Theoretical Physics, State
University of New York at Stony Brook, Stony Brook, NY 11794-3840,
USA }
\affiliation{Institute for Space Sciences,
 Bucharest-M\u{a}gurele, R 077125, Romania}

\author{Vladimir E. Korepin}

\affiliation {C.N. Yang Institute  for Theoretical Physics, State
University of New York at Stony Brook, Stony Brook, NY 11794-3840,
USA }

\author{Dmitri V. Averin}
\affiliation{Department of Physics and Astronomy, State University
of New York at Stony Brook, Stony Brook, NY 11794-3800, USA }
\email[Electronic addresses: ]{ipatu@grad.physics.sunysb.edu;
korepin@max2.physics.sunysb.edu ; dmitri.averin@stonybrook.edu}

\begin{abstract}

Using the determinant representation for the field-field correlation
functions of impenetrable anyons at finite temperature obtained in
\cite{PKA3}, we derive a system of nonlinear partial differential
equations completely characterizing the correlators. The system is
the same as the one for impenetrable bosons but with different
initial conditions. The large-distance asymptotic behavior of the
correlation functions is obtained from the analysis of the
Riemann-Hilbert problem associated with the system of differential
equations. We calculate both the exponential and pre-exponential
factors in the asymptotics of the field-field correlators. The
asymptotics derived in this way agree with those of the free fermions
and impenetrable bosons in the appropriate limits, $\kappa\rightarrow
1$ and $\kappa\rightarrow 0$, of the statistics parameter $\kappa$,
and coincide with the predictions of the conformal field theory at
low temperatures.

\end{abstract}

 \maketitle

\section{Introduction and Statement of Results }

This is the third paper in the series investigating the correlation
functions of one-dimensional impenetrable anyons at finite
temperatures. In the previous two papers \cite{PKA2,PKA3}, we have
obtained the anyonic generalization of Lenard's formula and
determinant representations for the time-, space-, and
temperature-dependent correlators of the anyons. Here, starting from
the obtained determinant representation, we derive the explicit
expressions for the large-distance asymptotics of the anyonic space
correlators at finite temperatures. These expressions demonstrate
the crossover between the fermionic and bosonic behavior of the
correlators when the appropriately-defined statistics parameter
$\kappa \in [0,1]$ changes between the corresponding limits of
bosons, $\kappa =0$, and fermions, $\kappa =1$. The approach we use
to compute the asymptotics of the correlation functions of
impenetrable anyons follows the one developed for impenetrable
bosons \cite{KBI,IIKS,IIK3}. The main technical point of this
approach is the large-distance asymptotic analysis of the
Riemann-Hilbert problem associated with the system of differential
equations which can be derived from the determinant representation.

The model of impenetrable one-dimensional (1D) anyons we consider is
the limit of infinite interaction strength of the Lieb-Liniger gas
of anyons with delta-functional repulsion \cite{SJBR,AGRPS,Kundu}.
The fact that impenetrable particles can not be directly exchanged
in 1D systems implies that introduction of the exchange statistics
for such particles requires an additional convention on the choice
of the sign of the statistical phase associated with exchange of
each pair of particles in the system wavefunctions \cite{AN,PKA}.
The absence of direct exchanges also implies that local
thermodynamic properties of $N$ impenetrable 1D anyons are
independent of statistics in the thermodynamic limit $N \rightarrow
\infty$, and coincide with those of the free fermions -- see, e.g.
\cite{BGH}. In the case of quasiperiodic boundary conditions
(particles moving on a circle), thermodynamic properties are
statistics-dependent at the level of corrections to the leading
large-$N$ asymptotic terms \cite{ZW,BGO,MDG,SSC}. Although the
non-local characteristics like single-particle momentum distribution
(i.e., Fourier transform of the single-particle density matrix)
depend strongly on the statistics parameter $\kappa$ \cite{SSC,C},
they may not be measurable in practice \cite{OP}. For transparent
anyons, i.e. finite particle-particle repulsion, local thermodynamic
properties do depend on $\kappa$ through modification of the
effective coupling constant \cite{BGO,PKA,CM,HZC}, but the model of
anyons without a hard core is not well-defined because of the
essential singularity of the system wavefunction at coincident
particle coordinates. Qualitatively, these properties of the
one-dimensional anyons are preserved in other related anyonic models
with, e.g., motion on the lattice \cite{AOE1,OAE1,BFGLZ,HZC1} or
different interactions \cite{BGK,CS}.

This work presents the exact calculation of the large-distance asymptotic
of the field-field correlator of impenetrable anyons. To state its
main result, we need some notations and definitions. First, we
introduce the function
\be\label{defnuf1} \nuf(\lam,\beta)=\frac{1}{2\pi i}\ln
\left(\frac{e^{\lam^2-\beta}+1 }{e^{\lam^2-\beta} -e^{i\pi\kappa}
}\right) , \ee
where the branch of the logarithm is specified by the requirements
that no branch cut intersects the real axis, and
\be\label{log} \lim_{\lam\rightarrow \infty} \ln
\left(\frac{e^{\lam^2-\beta} +1}{e^{\lam^2-\beta}
-e^{i\pi\kappa}}\right)\rightarrow 0 \, .  \ee
We then define two constants $C(\beta,\kappa)$ and $c(\beta,\kappa)$
as
\be \label{defcbk1} C(\beta,\kappa) = 2i\inti \nuf(\lam,\beta)
d\lam=\frac{1}{\pi}\inti \ln \left(\frac{e^{\lam^2-\beta}+1
}{e^{\lam^2-\beta}-e^{i\pi\kappa}}\right)d\lam\, , \ee
and
\be \label{con}  c(\beta,\kappa)=\inti\frac{\partial_\lam
\nuf(\lam,\beta) \nuf(\mu,\beta) -\nuf(\lam,\beta)
\partial_\mu\nuf(\mu,\beta)}{2 (\lam-\mu)}\ d\lam d\mu\, . \ee
In terms of these quantities, our main result for the two leading
terms in the large-distance asymptotic of the field-field correlator
of impenetrable anyons is:
\be\label{asymptotic}
\la\fad(x_1)\fa(x_2)\ra_T=e^{-x_{12}\sqrt{T}C(h/T,\kappa)/2}
e^{c(h/T, \kappa)} \left(c_0e^{ix_{12}\sqrt{T}\lam_0} +c_{-1}
e^{ix_{12}\sqrt{T} \lam_{-1}}\right)\, . \ee
Here $T$ is temperature, $h$ - chemical potential, $x_{12}\equiv
(x_1-x_2) \rightarrow +\infty $, and the constants $\lam_j$ with
$j=0, -1$ are complex numbers which depend on $\beta \equiv h/T$ and
statistics $\kappa$ as
\be
\lam_j=\left(\beta+\sqrt{\beta^2+\pi^2[\kappa+2j]^2}\right)^{1/2}/
\sqrt{2} +i\left(-\beta+\sqrt{\beta^2+\pi^2[\kappa+2j]^2}
\right)^{1/2} /\sqrt{2}\, . \ee
Two other constants $c_j$ which give the amplitudes of the
asymptotic terms in (\ref{asymptotic}) are given by
\be
c_j=i\frac{e^{i\pi\kappa}\sqrt{T}}{2}\frac{\alf^{-2}(\lam_j)}{\lam_j}\,
, \ee
where the function $\alf(\lam)$ is defined by Eq.~(\ref{defalf}).
The second term in the asymptotic expansion (\ref{asymptotic}) is
always smaller than the first one, but becomes relevant when the
statistics parameter approaches the fermionic value $\kappa = 1$.
The precise conditions on $\kappa$ under which this term does not
exceed the accuracy of our calculation are stated in Section
\ref{largexap}, see Eqs.~(\ref{min}) and (\ref{max}). More
generally, all results of this work are obtained, strictly speaking,
only for $\kappa\in (0,1]$. However, as shown in Section
\ref{bosoniclimit}, in the limit $\kappa\rightarrow 0$ our results
reproduce completely the exponential behavior of the asymptotics for
impenetrable bosons ($\kappa=0$) \cite{IIK2,IIK3}, and the
pre-exponential factors in the case of negative chemical potential
$h$.

An important role in obtaining the asymptotics (\ref{asymptotic}) is
played by a set of auxiliary potentials $B_{+-}$ and $B_{++}$
introduced below in Section \ref{DRAP} which are related to the
integral operators in the representation of the correlation
functions. The integral operator which enters in the Fredholm
determinant representation for the correlator is of a special kind
called ``integrable integral operator". Integral operators of this
type play an important role in the study of correlation functions of
integrable models and random matrices
\cite{WMTB,JMMS,IIKS,IIKV,HI,KBI}. The special form of the kernels of
such operators allowed us to obtain a system of partial differential
equations for the auxiliary potentials, and for the logarithm of the
Fredholm determinant, which characterize completely the correlators
at any distance. The system of equations is the same as the one
obtained for impenetrable bosons \cite{IIK1,IIK3,KBI} but with
different initial conditions. The short-distance and low-density
behavior of the field-field correlator can be extracted from the
initial conditions. At zero temperature, the nonlinear differential
equation for the logarithm of the determinant becomes the Painlev\'e
V differential equation obtained by  Jimbo, Miwa, M\^ori and Sato
for impenetrable bosons \cite{JMMS}, but again, with different
initial conditions. For finite system, the Painlev\' e VI
differential equation characterizing the field-field correlator was
derived by Santachiara and Calabrese in \cite{SC}. The
short-distance behavior of the 2-point correlation function obtained
in \cite{SC} agrees in the thermodynamic limit with our results at
zero temperature. The large-distance behavior of the correlation
functions is extracted from the solution of a matrix Riemann-Hilbert
problem associated with the system of differential equations for the
auxiliary potentials. The technique used is similar to the one
developed for impenetrable bosons \cite{IIK1,IIK3,KBI}. For the
computation of the pre-exponential factors in the asymptotics of the
correlators we have adopted the method employed by Kitanine,
Kozlowski, Maillet, Slavnov and Terras in their study of the
generalized sine-kernel \cite{KKMST}. This method was first
used by Cheianov and Zvonarev in the computation of the asymptotic
behavior  of correlation functions of impenetrable electrons \cite{CZ,CZ1}.
Some of the results of this work were presented in \cite{PKA4}.

The plan of the paper is as follows. In Section \ref{DRAP}, we
briefly review the determinant representation for correlators
obtained in our previous papers \cite{PKA2,PKA3}, and introduce the
auxiliary potentials which play a central  role in the subsequent
analysis of the correlation functions. In Section \ref{DESC}, we
derive the system of partial nonlinear differential equations
characterizing the auxiliary potentials and the correlators. Section
\ref{MRH} introduces the matrix Riemann-Hilbert problem associated
with the obtained system of equations. The large-distance analysis
of this problem is performed in Section \ref{LDAA}. We conclude in
Section \ref{A} by comparing our main result with the known limiting
cases: the low-temperature limit characterized by the conformal
behavior of the correlators, and bosonic and fermionic limits of the
statistics parameter. Technical details of the calculations are
presented in the appendices.

\section{Determinant Representation and Auxiliary Potentials}\label{DRAP}

The model of one-dimensional anyons we consider in this work is
characterized by the hamiltonian $H$ given in the second quantized
form as
\be\label{hama} H=\int dx  \left( [\partial_x \fad(x)][\partial_x
\fa(x)]+c\fad(x)\fad(x)\fa(x)\fa(x)-h\fad(x)\fa(x) \right) . \ee
Here $c$ is the coupling constant which is assumed to be large,
$c\rightarrow +\infty$, so that the anyons are impenetrable. The
fields in the Hamiltonian (\ref{hama}) obey the anyonic commutation
relations
\[ \fa(x_1)\fad(x_2)=e^{-i\pi\kappa\epsilon(x_1-x_2)} \fad(x_2)
\fa(x_1)+ \delta(x_1-x_2)\, , \]
\[ \fad(x_1)\fad(x_2)=e^{i\pi\kappa\epsilon(x_1-x_2)} \fad(x_2)
\fad(x_1)\, , \]
where $\epsilon(x)=x/|x|,\ \epsilon(0)=0.$ The commutation relations
become bosonic for $\kappa=0$, and fermionic for $\kappa=1$. In this
work, we are interested in the large-distance behavior of  the
field-field correlator of the impenetrable anyons (\ref{hama}). The
static, i.e. equal-time, correlator at finite temperatures is
defined by the standard relation:
\be \la\fad(x_1)\fa(x_2)\ra_T\equiv\frac{\mbox{Tr} \left(
e^{-H/T}\fad(x_1)\fa(x_2)\right)}{\mbox{Tr} e^{-H/T}}\, . \ee
In the previous paper \cite{PKA2} of the series, we have obtained
the following expressions for this correlator
\be \label{i7} \la\fad(x_1)\fa(x_2)\ra_{T}= \frac{1}{2\pi}\mbox{Tr
}\left[(1-\gamma\hat K_T)^{-1}\hat A_T^+\right] \det(1-\gamma\hat
K_T)|_{\gamma=(1+e^{+i\pi\kappa})/\pi}\, ,\ \ \  x_1>x_2\, , \ee
and
\be\label{i9} \la\fad(x_1)\fa(x_2)\ra_{T}= \frac{1}{2\pi}\mbox{Tr
}\left[(1-\gamma\hat K_T)^{-1}\hat A_T^-\right] \det(1-\gamma\hat
K_T)|_{\gamma=(1+e^{-i\pi\kappa})/\pi}\, ,\ \ \  x_1<x_2\, .
\ee
In these expressions, $\hat K_T$ and $\hat A_T^\pm$ are the integral
operators acting on the entire real axis with the kernels
\[ K_T(\lam,\mu)=\sqrt{\vartheta(\lam)}\frac{\sin x_{12}(\lam-\mu)}{
\lam -\mu}\sqrt{\vartheta(\mu)}\, ,\ \ \  \ \
 A_T^\pm(\lam,\mu)=\sqrt{\vartheta(\lam)}e^{\mp ix_{12}(\lam+\mu)}
\sqrt{\vartheta(\mu)}\, , \]
where
\be
\vartheta(\lam)\equiv\vartheta(\lam,T,h)=\frac{1}{1+e^{(\lam^2-h)/T}}\,
, \label{fw} \ee
is the Fermi distribution function, and $\mbox{Tr}[ f(x,y)] \equiv
\int f(x,x)\ dx$. Equations (\ref{i7}) and (\ref{i9}) remain valid
at zero temperature, but in this case, the integral operators act on
the interval $[-q,q]$ with $q=\sqrt{h}$, and have kernels
\[ K(\lam,\mu)=\frac{\sin x_{12}(\lam-\mu)}{\lam -\mu}\, ,\ \ \ \
A^\pm(\lam,\mu)=e^{\mp ix_{12}(\lam+\mu)}\, . \]

Equations (\ref{i7}) and (\ref{i9}) show that the anyonic
corrrelator $\la\fad(x_1)\fa(x_2)\ra_{T}$ depends on the sign of
$x_1-x_2$. However, since the values (\ref{i7}) and (\ref{i9}) for
positive and negative $x_1-x_2$ are related directly via complex
conjugation, one can focus only on one range, e.g., the correlator
(\ref{i7}). This correlator depends on four variables: the
coordinate difference $x_1-x_2>0$, temperature $T$, chemical
potential $h$, and statistics parameter $\kappa$. As we will see
bellow, introducing the variables that are rescaled by temperature:
the distance $x$ and the chemical potential $\beta$ defined as
\[x=\frac{1}{2}(x_1-x_2)\sqrt{T}\, , \ \ \ \beta=\frac{h}{T} \, ,
\]
and similarly changing the spectral parameter, $\lam\rightarrow \lam
\sqrt{T}$,  one makes explicit dependence on temperature very
simple:
\be\label{fc} \la\fad(x_1)\fa(x_2)\ra_{T}= \frac{\sqrt{T} }{2\pi
\gamma} g(x,\beta,\gamma)|_{\gamma=(1+e^{i\pi\kappa})/\pi}\, . \ee
The function $g(x,\beta,\gamma)$ here will be defined in the next
section.


\subsection{Auxiliary Potentials}

The Fredholm integral operator $\hat K_T$ appearing in the
expressions (\ref{i7}) and (\ref{i9}) for the field correlators
belongs to a special class of ``integrable'' operators
\cite{IIKS,KBI,HI}. This means that in terms of the ``plane waves"
introduced as
\be \label{de} e_{\pm}(\lam)=\sqrt{\vartheta(\lam)}e^{\pm i\lam x}
\ee
the kernel of $\hat K_T$ can be written in the following form which
generalizes the simple factorizable kernels:
\be \label{ad1} K_T (\lam,\mu)= \frac{e_+(\lam)e_-(\mu)-e_-
(\lam)e_+(\mu)}{ 2i (\lam -\mu)} \, . \ee
This is a particular case of the more general situation studied in
\cite{IIKS,HI} (see also Chap XIV of \cite{KBI}). An important
feature of this class of integrable operators is the fact that the
kernel $R_T(\lam,\mu)$ of the resolvent operator $\hat R_T$, defined
by the relation
\be \label{defrt} \hat R_T=(1-\gamma \hat K_T)^{-1}\hat K_T\, ,\ \
\mbox{i.e.},\  (1-\gamma\hat K_T)(1+\gamma\hat R_T)=1\, , \ee
has the same form \cite{IIKS,KBI}. Indeed, Eq.~(\ref{defrt}) means
that the resolvent kernel $R_T(\lam,\mu)$ solves the integral
equation
\[ R_T(\lam,\mu) -\gamma\int_{- \infty}^{+\infty}K_T(\lam,\nu)
R_T(\nu,\mu)\ d\nu=K_T(\lam,\mu)\, . \]
Introducing then the functions $f_\pm(\lam)$ which are the solutions
of the similar integral equations
\be\label{df}
f_\pm(\lam)-\gamma\int_{-\infty}^{+\infty}K_T(\lam,\mu) f_\pm(\mu)
d\mu=e_\pm(\lam)\, , \ee
one can show (for a proof, see Chap XIV of \cite{KBI}) that the
resolvent kernel can be written in the same form (\ref{ad1}) as
$K_T$:
\be\label{dr}
R_T(\lam,\mu)=\frac{f_+(\lam)f_-(\mu)-f_-(\lam)f_+(\mu)}{2i( \lam -
\mu)}\, . \ee

An important role in our asymptotic analysis of the correlator
(\ref{i7}) is played by the auxiliary potentials $B_{lm}$ which are
defined by
\be\label{defblm} B_{lm}(x,\beta,\kappa)\equiv
\gamma\int_{-\infty}^{+\infty} e_l(\lam)f_m(\lam)\ d\lam\, , \ \ \ \
l=\pm,\ m=\pm\, , \ee
where $\gamma=(1+e^{i\pi\kappa})/\pi$. A new feature of the
auxiliary potentials $B_{lm}$ in the case of anyons in comparison to
bosons, when $\gamma=2/\pi$, is that they are now complex. The
potentials, however, still satisfy the relations $B_{+-}(x, \beta,
\kappa) =B_{-+}(x,\beta,\kappa)$ and $B_{++}(x,\beta,\kappa) =
B_{--}(x,\beta,\kappa)$ as in the bosonic case. Indeed, as one can
see from the definitions (\ref{defblm}) and (\ref{df})
\begin{eqnarray*}
B_{+-}& = &\gamma\int_{-\infty}^{+\infty} e_+(\lam)f_-(\lam)\ d
\lam= \gamma\int_{-\infty}^{+\infty} e_+(\lam)\int_{-\infty}^{+
\infty}(1-\gamma \hat K_T)^{-1}(\lam,\mu)e_-(\mu)\ d\mu\  d\lam
\nonumber\\ &=&\gamma\int_{-\infty}^{+\infty} e_-(\lam)f_+(\lam)\
d\lam\equiv B_{-+}\, , \end{eqnarray*}
where in the last line we have used the fact that the kernel $K_T$
is symmetric, $K_T(\lam,\mu)=K_T(\mu,\lam)$. In order to prove the
second assertion we start with the integral equation (\ref{df}) for
$f_+(-\lam)$
\[ f_+(-\lam)-\gamma\int_{-\infty}^{+\infty}K_T(-\lam,\mu)f_+(\mu)
d\mu=\sqrt{\vartheta(\lam)}e^{-i\lam x}\, , \]
which can be rewritten as
\[ f_+(-\lam)-\gamma\int_{-\infty}^{+\infty}K_T(-\lam,-\mu)f_+(-\mu)
d\mu =\sqrt{\vartheta(\lam)}e^{-i\lam x}\, . \]
Using the relation $K_T(-\lam,-\mu)=K_T(\lam,\mu)$, we see then that
$f_+(-\lam)=f_-(\lam)$. This gives
\begin{eqnarray*}
B_{++}& = &\gamma\int_{-\infty}^{+\infty} e_+(\lam)f_+(\lam)\ d\lam
=\gamma\int_{-\infty}^{+\infty} e_+(-\lam)f_+(-\lam)\ d\lam \nonumber
\\ &=&\gamma\int_{-\infty}^{+\infty} e_-(\lam)f_-(\lam)\ d\lam =
B_{--} \, . \end{eqnarray*}
Finally, the definition of the kernel $A_T^+$ shows directly that
(\ref{i7}) can be rewritten as
\[ \la\fad(x_1)\fa(x_2)\ra_{T}= \frac{\sqrt{T}}{2\pi}\det(1-\gamma
\hat K_T)|_{\gamma=(1+e^{+i\pi\kappa})/\pi} \int_{-\infty}^{+\infty}
f_-(\lam)e_-(\lam)\ d\lam \, , \]
in terms of the function $f_-(\lam)$ (\ref{df}). This shows that the
function $g(x,\beta,\gamma)$ introduced in Eq.~(\ref{fc}) is given
by
\be\label{dg} g(x,\beta,\gamma)= B_{++}(x,\beta,\gamma) \det(1-
\gamma\hat K_T)|_{\gamma=(1+e^{+i\pi\kappa})/\pi} \, . \ee
%

\section{Differential Equations for Static Correlators}\label{DESC}

In general, it is difficult to obtain directly differential
equations for the entire correlator (\ref{i7}). The strategy we
pursue is to obtain first nonlinear partial differential equations
for the potentials $B_{++},B_{+-}$, and then show that the function
$\sigma(x,\beta, \gamma) = \ln \det(1-\gamma \hat K_T)$ can be
expressed in terms of $B_{++}$ and $B_{+-}$. We start by look for
the two operators \textsf{L} and \textsf{M} depending on $B_{++}$
and $B_{+-}$ such that we have
\begin{eqnarray}\label{Lax}
\partial_x F(\lam)&=&\textsf{L}\  F(\lam)\, ,\nonumber\\
(2\lam\partial_\beta+\partial_\lam) F(\lam)&=&\textsf{M}\ F(\lam)\, ,
\end{eqnarray}
where $F(\lam)$ is the two-component vector function
\be \label{F} F(\lam)=\left(\begin{array}{c} f_+(\lam)\\f_-(\lam)
\end{array} \right)\, . \ee
The compatibility condition for the system (\ref{Lax})
\[ (2\lam\partial_\beta+\partial_\lam)\textsf{L}-\partial_x
\textsf{M}+ [\textsf{L},\textsf{M}]=0\, , \]
gives then a system of partial differential equations in $x$ and
$\beta$. It should be noted that since the Fredholm determinant in
Eq.~(\ref{i7}) differs from the one appearing in the similar
representation for impenetrable bosons only in the value of
$\gamma$, the calculations follow closely those for impenetrable
bosons \cite{IIK1,KBI}. The main difference is that $\gamma$ for
anyons is complex, making the auxiliary potentials $B_{++}, B_{+-}$
complex and not real as for the bosons.


\subsection{The \textsf{L} operator}

We start from the integral equations for the functions $f_\pm(\lam)$
\[ f_\pm(\lam)-\gamma\int_{-\infty}^{+\infty}K_T(\lam,\mu)f_\pm(\mu)
d\mu=e_\pm(\lam)\, ,\ \ \gamma=(1+e^{i\pi\kappa})/\pi. \]
Differentiation with respect to $x$ gives
\[ \left[(1-\gamma\hat K_T)\partial_x f_\pm\right](\lam)-\gamma
\left[ \partial_x\hat K_T f_\pm\right](\lam)=\pm i\lam e_\pm(\lam)
\, . \]
Acting on both sides with $(1-\gamma \hat K_T)^{-1}$  and using the relations
$(1-\gamma\hat K_T)^{-1}=(1+\gamma\hat R_T),$ $\left[(1-\gamma \hat K_T)^{-1}
e_\pm\right](\lam)=f_\pm(\lam),$ $\partial_x
K_T(\lam,\mu)=(e_+(\lam)e_-(\mu)+e_-(\lam)e_+(\mu))/2$, and the
special factorization for the resolvent kernel (\ref{dr})  we obtain
\begin{eqnarray*}
\partial_x f_+(\lam)&=&i\lam f_+(\lam)+B_{++}f_-(\lam)\, ,\\
\partial_x f_-(\lam)&=&-i\lam f_+(\lam)+B_{--}f_-(\lam)\, .
\end{eqnarray*}
These equations mean that the \textsf{L} operator has the form
\be\label{defL}
\textsf{L}=i\lam\sigma_3+B_{++}\sigma_1\, ,
\ee
where we have used the fact that $B_{++}=B_{--}$, and $\sigma_i$ are
the Pauli matrices
\[ \sigma_1=\left(\begin{array}{lr}0&1\\1&0 \end{array}\right), \ \
\ \sigma_2=\left(\begin{array}{lr}0&-i\\i&0 \end{array}\right), \ \
\ \sigma_3=\left(\begin{array}{lr}1&0\\0&-1 \end{array}\right). \]
%


\subsection{The \textsf{M} operator}

The derivation of the \textsf{M} operator is more complicated. In
this case, we have to rely heavily on the following property of the
Fermi distribution function $\vartheta(\lam)$:
\be \label{ad5} (2\lam\partial_\beta+\partial_\lam)\vartheta(\lam)
=0 \, . \ee
This property is essential, since in the derivation of the
\textsf{M} operator, one encounters terms which contain the
resolvent $R_T(\lam,\mu)$ but still can not be reduced to the form
$(\lam-\mu)R_T(\lam,\mu)$ (as in the previous section) which produces
``disentangled" terms (i.e., products of the one-dimensional
projectors). This is also why the differential operator associated
with \textsf{M} is $(2\lam\partial_\beta+\partial_\lam)$ and not
simply $\partial_\beta$ as one could have expected.

The computations are more involved but similar with the ones for impenetrable
bosons and we refer the reader to  Chap. XIV of \cite{KBI} for the derivation.
The final result is
\begin{eqnarray*}
(2\lam\partial_\beta+\partial_\lam)f_+(\lam)&=&ixf_+(\lam)-if_+(\lam)
\partial_\beta B_{-+}+if_-(\lam)\partial_\beta B_{++} \, ,\\
(2\lam\partial_\beta+\partial_\lam)f_-(\lam)&=&-ixf_+(\lam)-if_+(\lam)
\partial_\beta B_{--}+if_-(\lam)\partial_\beta B_{+-}\, .
\end{eqnarray*}
Taking into account that $B_{-+}=B_{+-},\ B_{++}=B_{--}$ this
finally gives the \textsf{M} operator as
\be\label{defM} \textsf{M}=ix\sigma_3-i\partial_\beta
B_{+-}\sigma_3- \partial_\beta B_{++} \sigma_2\, . \ee
%


\subsection{Differential equations for the potentials} \label{sigma}

The results obtained in the previous sections allow us to state the
following theorem:

\begin{thm}\label{thm1}
For all $\gamma=(1+e^{i\pi\kappa})/\pi$ with $\kappa\in[0,1)$, the
potentials $B_{+-}(x,\beta,\gamma)\, , B_{++}(x,\beta,\gamma)$
satisfy the following system of partial differential equations
\begin{eqnarray}
\partial_\beta B_{+-}&=&x+\frac{1}{2}\frac{\partial_x\partial_\beta
B_{++}}{B_{++}}\, ,\label{debpm}\\ \partial_x B_{+-}&=&B_{++}^2\,
,\label{debpp} \end{eqnarray}
with the initial conditions (at fixed $\beta$)
\be \label{inx} B_{++}(x,\beta,\gamma)=B_{+-}(x,\beta,\gamma) =
\gamma d(\beta)+ [\gamma d(\beta)]^2x+O(x^2)\, ,\ \ x\rightarrow 0\,
, \ee
where $d(\beta)=\inti\vartheta(\lam)d\lam$ and
\be\label{inbeta} B_{++}(x,\beta,\gamma)=B_{+-}(x,\beta,\gamma)=0\,
,\ \ \beta\rightarrow -\infty\, . \ee
The  potential $B_{++}(x,\beta,\gamma)$ satisfies the nonlinear
equation
\be\label{debpp2}
\partial_\beta B_{++}^2=1+\frac{1}{2}\frac{\partial}{\partial_x}
\left(\frac{\partial_x\partial_\beta B_{++}}{B_{++}}\right)\, ,
\ee
with the same initial conditions.
\end{thm}

\begin{proof}
In the previous sections we have shown that the two-component vector
function $F(\lam)$ satisfies the following differential equations
\begin{eqnarray*}
\partial_x F(\lam)&=&\textsf{L}\  F(\lam)\, ,\\
(2\lam\partial_\beta+\partial_\lam) F(\lam)&=&\textsf{M}\ F(\lam)\, .
\end{eqnarray*}
where \textsf{L,M} are given by (\ref{defL}) and (\ref{defM}). The
compatibility condition for these equation is
\be\label{compcond}
[\partial_x-\textsf{L},(2\lam\partial_\beta+\partial_\lam)-
\textsf{M} ] = (2\lam\partial_\beta+\partial_\lam)\textsf{L}-
\partial_x \textsf{M} + [\textsf{L},\textsf{M}]=0\, . \ee
Differentiating Eqs.~(\ref{defL}) and (\ref{defM}), we get
\[ (2\lam\partial_\beta+\partial_\lam)\textsf{L}=2\lam
(\partial_\beta B_{++})\sigma_1+i\sigma_3\, , \]
\[ \partial_x\textsf{M}=i\sigma_3-i(\partial_x\partial_\beta B_{+-}
)\sigma_3-(\partial_x\partial_\beta B_{++})\sigma_2\, .  \]
Using the standard relation for the Pauli matrices
\[ [\sigma_i,\sigma_j]=2i\epsilon_{ijk}\sigma_k\, , \]
we also have
\[ [\textsf{L},\textsf{M}]=-2\lam(\partial_\beta B_{++})\sigma_1+
\left[2x B_{++}-2B_{++}(\partial_\beta B_{+-})\right]\sigma_2
-2iB_{++}(\partial_\beta B_{++})\sigma_3\, . \]
Plugging these results into Eq.~(\ref{compcond}) we obtain a matrix
with only two distinct elements. The condition that this matrix
vanishes identically gives the following relations
\[ 2B_{++}(\partial_\beta B_{+-})-2x B_{++}+\partial_x
\partial_\beta B_{++}=0\, , \]
and
\[ \partial_x\partial_\beta B_{+-}-\partial_\beta B_{++}^2=0\, . \]
The first relation is just Eq.~(\ref{debpm}), while the second is
equivalent to Eq.~(\ref{debpp}) due to (\ref{inbeta}). The equation
(\ref{debpp2}) is obtained by differentiating (\ref{debpm}) with
respect to $x$, (\ref{debpp}) with respect to $\beta$, and then
equating the right-hand sides of the resulting relations. The
asymptotic behavior at short distances (\ref{inx}) and at low
density (\ref{inbeta}), which give the initial conditions, are
obtained in Appendices (\ref{sda}) and (\ref{lde}).
\end{proof}


\subsection{Differential equations for $\sigma$}

\begin{thm}\label{thm2}
For any $\gamma=(1+e^{i\pi\kappa})/\pi$ with $\kappa\in[0,1)$, the
partial derivatives of $\sigma(x,\beta,\gamma) \equiv \ln
\det(1-\gamma \hat K_T)$ with respect to $x$ and $\beta$ are given
by
\begin{eqnarray}
\partial_x\sigma&=&-B_{+-}\, , \ \ \ \partial_x^2\sigma=-B_{++}^2
\, ,\label{derxs}\\ \partial_\beta\sigma &=&-x\partial_\beta
B_{+-}+\frac{1}{2}(\partial_\beta B_{+-})^2-\frac{1}{2}
(\partial_\beta B_{++})^2\, .\label{derbs}
\end{eqnarray}
Furthermore,  for all $\gamma$s, the function $\sigma(x,\beta,
\gamma)$ satisfies the following nonlinear partial differential
equation
\be\label{mainde} (\partial_\beta\partial_x^2 \sigma)^2+4
(\partial_x^2 \sigma)[ 2x\partial_\beta\partial_x\sigma+
(\partial_\beta\partial_x\sigma)^2-2\partial_\beta\sigma]=0 \ee
with the initial conditions
\be\label{incs} \sigma=-\gamma d(\beta)x-[\gamma
d(\beta)]^2\frac{x^2}{2}+O(x^3)\, , \ \ x\rightarrow 0\, ;\ \
\sigma=0\, ,\ \  \beta\rightarrow -\infty\, , \ee
where, as before, $d(\beta)=\inti\vartheta(\lam)d\lam$.
\end{thm}

\begin{proof}
We start with the simpler case of derivative with respect to $x$.
Making use of the following standard representation for the Fredholm
determinant of an arbitrary integral operator
\be\label{repr} \det(1-\gamma \hat K_T)=\exp\left(- \sum_{n=1
}^\infty \frac{\gamma^n}{n}\mbox{ Tr }K_T^n\right) , \ee
where $\mbox{Tr } K_T=\int K_T(\lam,\lam),\ \ \mbox{Tr } K_T^2= \int
K_T(\lam,\mu)K_T(\mu,\lam)\  d\lam d\mu $, and so on,  we obtain
\[ \partial_x\sigma=-\gamma \mbox{ Tr }\left[ (1-\gamma \hat K_T)^{
-1} \partial_x \hat K_T\right]\, . \]
From the relations $\left[(1-\gamma \hat K_T)^{-1} e_\pm\right]
(\lam )=f_\pm(\lam)$ and $\partial_x K_T(\lam,\mu)=(e_+(\lam)
e_-(\mu) +e_-(\lam)e_+(\mu))/2$ that were already used before, we
have
\begin{eqnarray*}
-\gamma \mbox{ Tr }\left[(1-\gamma \hat K_T)^{-1}\partial_x \hat
K_T\right]&=& -\frac{\gamma}{2}\inti f_+(\lam)e_-(\lam)d\lam-
\frac{\gamma}{2}\inti f_-(\lam)e_+(\lam)d\lam\, ,\nonumber\\ &=&
-\frac{1}{2}(B_{-+} +B_{+-})\, , \end{eqnarray*}
which due to $B_{-+}=B_{+-}$ means that
\[ \partial_x \sigma=-B_{+-}\, . \]
proving the first part of Eq.~(\ref{derxs}). The second part is
obtained differentiating once more with respect to $x$ and using
Eq.~(\ref{debpp}).

The derivation of the differential equation (\ref{derbs}) involving the  $\beta$ derivative
is more cumbersome. Again, the computations are similar
with the impenetrable bosons case and we refer the reader to Chap. XIV of \cite{KBI}.

The last differential equation (\ref{mainde}) is obtained by
replacing the potentials with $B_{+-}=-\partial_x \sigma$ and
$B_{++}=(-\partial_x^2\sigma)^{1/2}$ in the R.H.S. of
Eq.~(\ref{derbs}). It is also straightforward to see that $K_T (
\lam,\mu) \rightarrow 0$ when $x\rightarrow 0$ or $\beta \rightarrow
-\infty$, which means that $\sigma= \ln \det (1-\gamma \hat K_T)=0$
in these limits. The first condition in (\ref{incs}) follows then
from Eq.~(\ref{inx}). This concludes the proof.
\end{proof}

Finally, we end this section with a simple but important
observation. Integrating equation $\partial_x\sigma=-B_{+-}$ with
initial conditions (\ref{incs}), we express the function
$g(x,\beta,\gamma)$ (\ref{dg}) and, therefore, the field correlator
(\ref{fc}), directly in terms of the potentials $B_{++},B_{+-}$:
\be\label{fullexp} g(x,\beta,\gamma) =B_{++}(x,\beta,\gamma) e^{-
\int_0^x B_{+-} (y,\beta,\gamma)dy} \, . \ee
This expression will be used in our subsequent analysis of the field
correlator.


\subsection{The zero-temperature limit}

In the zero temperature limit, the field-field correlator (\ref{fc})
depends only on three variables, the distance $(x_1-x_2)>0$, the
Fermi momentum  $q=\sqrt{h}$ (or, equivalently, the chemical
potential $h$), and the statistics parameter $\kappa$. The distance
and momentum dependence can be encoded by one variable:
\[ \xi=\frac{(x_1-x_2)}{2}q=x\beta^{1/2}\, , \]
where now the rescaled variables are $x=(x_1-x_2)/2$ and $\beta=h$.
At $T=0$, the integral operator $\hat K$ acts on the interval
$[-q,q]$ and has the kernel $K(\lam,\mu)=\sin x(\lam-\mu) /(\lam
-\mu)$. The logarithm of the determinant is
\[ \tilde\sigma_0(\xi,\gamma)=\ln \det(1-\gamma\hat K)|_{\gamma=(1+
e^{i\pi\kappa})/\pi}\, .  \]
The partial differential equation (\ref{mainde}) characterizing
$\sigma(x,\beta,\gamma)$ becomes then an ordinary differential
equation in $\xi$. Introducing the new function
\be
\sigma_0=\xi(\tilde\sigma_0)'
\ee
where prime denotes the derivative with respect to $\xi$, we
transform Eq.~(\ref{mainde}) into
\be\label{JMO}
(\xi\sigma_0'')^2+4(\xi\sigma_0'-\sigma_0)[4\xi\sigma_0-4\sigma_0 +
(\sigma_0')^2]=0 \ee
with the boundary conditions
\be\label{inc0} \sigma_0=-2\gamma\xi-4\gamma^2\xi^2+O(\xi^3)\, , \ \
\gamma=(1+e^{i \pi\kappa})/\pi\, . \ee
The ordinary differential equation (\ref{JMO}) is the same Painlev\'e
V  equation obtained by Jimbo, Miwa, M\^ori and  Sato in their
celebrated work on the one-particle reduced density matrix
(field-field correlator) of impenetrable bosons \cite{JMMS}. The only
difference we find in the anyonic case is the boundary conditions
(\ref{inc0}) which, unlike the differential equation itself, depend
on the statistics parameter. In a certain sense, this result could be
expected, since it was noticed already by Jimbo et {\it al.} that the
same equation (\ref{JMO}), but with different boundary conditions,
characterizes the density matrix of both the impenetrable bosons and
the free fermions. Similar situation was also noted by Forrester,
Frankel, Garoni and Witte \cite{FFGW} in their study of systems of
finite number of particles with periodic boundary conditions. In this
case, the reduced density matrix satisfies a Painlev\'e VI
differential equation. Their work was recently extended to
impenetrable anyons by Santachiara and Calabrese \cite{SC}. At $T=0$,
our results agree with those of Santachiara and Calabrese. Indeed,
Eq.~(\ref{sdasympt}) for the first terms of the short-distance
expansion of the field-field correlator, which reflects the boundary
conditions (\ref{inc0}), reduces at $T=0$ to:
\be
\la\fad(x_1)\fa(x_2)\ra_0=D_0\left(1-\frac{\pi^2}{6}D_0^2(x_1-x_2)^2
+\gamma\frac{\pi^3}{18}D_0^3(x_1-x_2)^3\right)+O\left((x_1-x_2)^4
\right)\, , \ \ \ \gamma=(1+e^{i\pi\kappa})/\pi \, , \ee
where $D_0=q/\pi$ is the density at $T=0$. This result coincides with
the expansion obtained in \cite{SC} -- see Eq.~(47) of that work and
note that the definition of the statistical parameter $\kappa$ there
differs in sign from ours.

\section{The Matrix Riemann-Hilbert problem for the field-field
corelator} \label{MRH}

The most difficult step in the analysis of the field-field
correlation function is the calculation of its large-distance
asymptotics. As we have seen in the previous sections, the
correlation function is characterized by a completely integrable
system of nonlinear partial differential equations -- see Thm.
\ref{thm1}, (\ref{fc}) and (\ref{fullexp}). A powerful method of
analyzing these differential equations is the matrix Riemann-Hilbert
problem (RHP) formalism \cite{TF}. The solution of the associated matrix
RHP will allow us to obtain the large-distance asymptotics of the
potentials $B_{+-},B_{++}$, and therefore, the large-distance
asymptotics of the correlator.

We start by formulating the matrix Riemann-Hilbert problem relevant
to the case of potentials $B_{+-},B_{++}$. In what follows,
$\kappa\in(0,1]$. As we will see below, in this case we need to find
a $2\times 2$ matrix-valued function $\chi(\lam)$, nonsingular for
all $\lam \in\mathbb{C}$ and analytic separately in the upper and
lower half-plane, which is equal to the unit matrix at $\lam=\infty$
\[ \chi(\infty)=I\, , \]
and has the boundary values on the real axis which satisfy the
condition
\begin{eqnarray}\label{RHP}
\chi_-(\lam)&=&\chi_+(\lam)G(\lam)\, ,\ \ \ \ \chi_{\pm}(\lam)=
\lim_{ \epsilon\rightarrow 0^+}\chi(\lam\pm i\epsilon)\ \
\lam\in\mathbb{R}\, . \end{eqnarray}
The matrix $G(\lam)$ is called the conjugation matrix associated
with the RHP and is defined only for real $\lam$. In our case, it
has the form
\be \label{conj} G(\lam)=\left(\begin{array}{cc} 1+\pi\gamma
e_+(\lam)e_-(\lam)\, , & -\pi\gamma e_+^2(\lam) \\ \pi\gamma
e_-^2(\lam)\, , & 1-\pi\gamma e_+(\lam) e_-(\lam) \end{array}\right)
\ee
where all notations are the same as before, e.g.,
$e_\pm(\lam)=\sqrt{\vartheta(\lam)}e^{\pm i\lam x}$, etc. This means
that both $G(\lam)$ and $\chi(\lam)$ depend also on $x,\beta$ and
$\kappa$, but this dependence will be suppressed in our notations.


\subsection{Relation to the auxiliary potentials}\label{connect}

In this section, we show how the auxiliary potentials $B_{+-},
B_{++}$ are obtained from the solution of the RHP (\ref{RHP}). The
normalization condition $\chi(\infty)=I$ means that the solution of
the RHP has the following expansion for large $\lambda$
\be\label{lelam} \chi(\lam)=I+\frac{\Psi_1}{\lam}+ O\left( \frac{
1}{\lam^2}\right) , \ee
where $\Psi_1(x,\beta)$ is a $2\times 2$ matrix which depends only
on $\beta$ and $x$. Then, the following symmetry of the conjugation
matrix
\be G(\lam)=\sigma_1 G^{-1}(-\lam)\sigma_1\, , \ee
that can be checked explicitly from the definition (\ref{conj}),
determines the structure of $\Psi_1$. This matrix can be written as
\be\label{psi1} \Psi_1=\frac{1}{2i}\left(\begin{array}{cc}
B_{+-}&-B_{++}\\ B_{++}&-B_{+-} \end{array}\right)\, , \ee
where the factor $1/2i$ is introduced for convenience, and at the
moment we do not know yet that $B_{++}$ and $B_{+-}$ in
Eq.~(\ref{psi1}) are the auxiliary potentials (\ref{defblm}). To
prove this, we use the following formulation of the RHP. One can
show directly (see, e.g, Chap.~XV of \cite{KBI}) that the matrix RHP
is equivalent to the system of singular integral equations:
\be \label{inteq} \chi_+(\lam)=I+\frac{1}{2\pi i}\inti\frac{ \chi_+
(\mu)[I-G(\mu)]}{\mu-\lam-i0}d\mu\, , \ \ \lam\in\mathbb{R}\, . \ee
If we define
\be\label{deftild} \tilde\chi(\lam)=\chi_+(\lam)E(\lam)\, , \ee
where $E(\lam)$ is the triangular matrix
\[ E(\lam)=\left(\begin{array}{lr}
                         1&e_+(\lam)\\
                         0&e_-(\lam)
               \end{array}\right) , \]
the system of integral equations (\ref{inteq}) can be rewritten as
\be\label{int11} \tilde\chi(\lam)=E(\lam)+\frac{1}{2\pi
i}\inti\frac{\tilde\chi(\mu) \tilde G(\mu,\lam)}{\mu-\lam-i0}d\mu\,
,\ \ \lam\in\mathbb{R}\, , \ee
with
\[ \tilde G(\lam,\mu)=E^{-1}(\mu)[I-G(\mu)]E(\lam)=
\left(\begin{array}{cc}
         0&0\\
         -\pi\gamma e_-(\mu) & \pi\gamma(e_+(\mu)e_-(\lam)-e_+(\lam)
         e_-(\mu)) \end{array}\right) . \]
For the matrix elements, Eq.~(\ref{int11}) gives the integral
equations:
\begin{eqnarray*}
\tilde\chi_{12}&=&e_+(\lam)+\gamma\inti K_T(\lam,\mu)
\tilde\chi_{12} (\mu)d\mu\, ,\\
\tilde\chi_{22}&=&e_-(\lam)+\gamma\inti K_T(\lam,\mu)
\tilde\chi_{22}(\mu)d\mu\, , \end{eqnarray*}
with the same kernel (\ref{ad1}) we considered before:
\[ K_T(\lam,\mu)=\frac{e_+(\lam)e_-(\mu)-e_-(\lam)e_+(\mu)}{
2i(\lam-\mu)}\, , \]
and
\begin{eqnarray*}
\tilde\chi_{11}(\lam)&=&1-\frac{\gamma}{2i}\inti
\frac{\tilde\chi_{12}(\mu)e_-(\mu)}{\mu-\lam-i0}d\mu\, ,\\
\tilde\chi_{21}(\lam)&=&-\frac{\gamma}{2i}\inti
\frac{\tilde\chi_{22}(\mu)e_+(\mu)}{\mu-\lam-i0}d\mu\, .
\end{eqnarray*}
In the first two equations, the nonsingular character of the kernel
$K_T$ makes it possible to neglect the infinitesimal shift in the
denominator. These equations are just the integral equations
(\ref{df}) defining $f_\pm(\lam)$. Therefore, we can make the
following identification
\be\label{chif} \tilde\chi_{12}(\lam)=f_+(\lam)\, ,\ \ \
\tilde\chi_{22}(\lam)=f_-(\lam)\, . \ee
Also, noting from Eq.~(\ref{deftild}) that $\tilde \chi_{11}(\lam)
=\chi_{ 11,+}(\lam)$ and $\tilde\chi_{21}(\lam)=\chi_{22,+}(\lam)$,
and are analytical in the upper half plane of $\lam$, we have
\begin{eqnarray} \label{chi11}
\chi_{11}(\lam)&=&1-\frac{\gamma}{2i}\inti\frac{f_+(\mu)e_-(\mu)}{
\mu-\lam}d\mu\, ,\ \ \Imm \lam>0\, ,\nonumber\\
\chi_{21}(\lam)&=&-\frac{\gamma}{2i}\inti\frac{f_-(\mu)e_-(\mu)}{
\mu-\lam}d\mu\, ,\ \ \Imm \lam>0\, . \end{eqnarray}
Taking the limit $\lam\rightarrow\infty$ in these equations, we
obtain
\begin{eqnarray*}
\chi_{11}(\lam)&=&1+\frac{\gamma}{2i\lam }\inti f_+(\mu)e_- (\mu)
d\mu+O\left(\frac{1}{\lam^2}\right) ,\\
\chi_{21}(\lam)&=&\frac{\gamma}{2i\lam }\inti f_-(\mu)e_-(\mu)
d\mu+O\left(\frac{1}{\lam^2}\right) . \end{eqnarray*}
which shows that the components of the matrix $\Psi_1$ are indeed
the potentials (\ref{defblm}), as assumed in Eq.~(\ref{psi1}). As a
byproduct of this derivation, one can also see from Eq.~(\ref{chif})
that a part of the matrix product (\ref{deftild}) means that
\be \label{obtf} \left(\begin{array}{c} f_+(\lam)\\
f_-(\lam)\end{array}\right)= \chi_+(\lam)\left(\begin{array}{c}
e_+(\lam)\\ e_-(\lam)\end{array}\right) . \ee
This relation will be used in Sec.~\ref{LDAA} in the analysis of the
large-distance asymptotics.

\subsection{An useful transformation of the RHP} \label{trans}

As we have shown above, the potentials $B_{++},B_{+-}$ that
completely characterize the field-field correlator (\ref{fc}) through
Eq.~(\ref{fullexp}), can be obtained from the expansion of the
solution of the matrix RHP (\ref{RHP}). This means that the
large-distance asymptotics of the potentials, and hence of the field
correlator, can be obtained from the solution of the RHP. To find
this solution, we will use equivalent formulation of the RHP in terms
of the singular integral equation (\ref{inteq}). As a first step, we
transform the matrix RHP (\ref{RHP}) into a new form which has
conjugation matrix with $1$ on the diagonal. This is achieved by
introducing the matrix function $\Phi$ defined as
\[ \Phi(\lam)=\chi(\lam)\left(\begin{array}{cc}
                                         \bet^{-1}(\lam)&0\\
                                         0&\alf^{-1}(\lam)
                            \end{array}\right) , \]
where the functions $\alf(\lam),\bet(\lam)$ are:
\be\label{defalf} \alf(\lam)=\exp\left\{-\frac{1}{2\pi i}\inti
\frac{d\mu}{\mu-\lam} \ln\left(\frac{e^{\mu^2-\beta}-e^{i\pi \kappa}
}{ e^{\mu^2-\beta}+1}\right)\right\} , \ee
and
\[ \bet(\lam)=\exp\left\{-\frac{1}{2\pi i}\inti \frac{d\mu}{\mu-
\lam}\ln\left(\frac{e^{\mu^2-\beta}+e^{i\pi\kappa}+2}{e^{\mu^2-
\beta} +1} \right)\right\} . \]
The branch of the logarithm in these equations is specified by the
requirement that $\ln \, (...) \rightarrow 0$ for $\mu \rightarrow
\infty$, and it is important at this point that $\kappa\in(0,1]$.
Both functions $\alf(\lam)$ and $\bet(\lam)$ are analytic separately
in the upper and lower half-plane, and are the solutions of the
following scalar Riemann-Hilbert problems (see Appendix \ref{srhp})
\be \label{RHa} \alf_-(\lam)= \alf_+(\lam)g_\alf(\lam)\, ,\
\lam\in\mathbb{R}\, ; \;\;\; \alf(\infty)=1\, ; \;\;\;
g_\alf(\lam)=\frac{e^{\lam^2-\beta}-e^{i\pi\kappa}}{e^{\lam^2-
\beta} +1} \, , \ee
and
\be \label{RHb} \bet_-(\lam)=\bet_+(\lam)g_\bet(\lam)\, ,\
\lam\in\mathbb{R}\, ; \;\;\; \bet(\infty)= 1\, ; \;\;\;
g_\bet(\lam)=\frac{e^{\mu^2-\beta}+e^{i\pi\kappa}+2}{e^{\mu^2-\beta}+1}\,
. \ee
They also have the property
\[ \alf^{-1}(\lam)=\alf(-\lam)\, ,\ \ \bet^{-1}(\lam)=\bet(-\lam)\,
. \]
Using these relations one can see from Eq.~(\ref{RHP}) that the
matrix $\Phi$ satisfies the following RHP
\be \label{RHP2} \Phi_-(\lam)=\Phi_+(\lam)G_\Phi(\lam)\, , \ \
\lam\in\mathbb{R}\, ; \;\;\; \Phi(\infty) = 1\, , \ee
with the conjugation matrix
\be G_\Phi(\lam)=\left(\begin{array}{cc}
1&-\pi\gamma\vartheta(\lam)\bet_+(\lam)\alf_-^{-1}(\lam)e^{2i\lam x}
\\ \pi\gamma\vartheta(\lam)\alf_+(\lam)\bet_-^{-1}(\lam) e^{-2i\lam
x}& 1 \end{array}\right)\, .  \label{off} \ee
Equations (\ref{lelam}) and (\ref{psi1}), together with the fact
that $\alf(\lam)\, ,\bet(\lam) \rightarrow 1$ at $\lambda
\rightarrow \infty$, show that in terms of the components of the
matrix $\Phi$ that solves the transformed RHP (\ref{RHP2}), the
potentials $B_{++},B_{-+}$ are expressed as
\be \label{extract}
B_{+-}=2i\lim_{\lam\rightarrow\infty}\lam\left[\Phi_{11}(\lam)-
\bet^{-1}(\lam)\right]\, , \;\;\; B_{++}=-2i\lim_{\lam\rightarrow
\infty}\lam\Phi_{12}(\lam)\, . \ee
The RHP (\ref{RHP2}) and formulae (\ref{extract}) will be the basis
for our analysis of the large-distance asymptotics of the
field-field correlator of impenetrable anyons.

\section{Large-Distance Asymptotic Analysis}\label{LDAA}

In this section, we perform the large-distance asymptotic analysis
of the RHP (\ref{RHP2}). The strategy is to use Eq.~(\ref{extract})
to obtain the auxiliary potentials $B_{++}(x,\beta,\kappa)$ and
$B_{+-}(x,\beta,\kappa)$ from the large-$\lam$ expansion of
$\Phi_{11,+}(\lam)$ and $\Phi_{12,+}(\lam)$. Then, the function
$\sigma(x,\beta,\kappa)$ can be calculated using the differential
equations (\ref{derxs}) and (\ref{derbs}).

\subsection{The auxiliary potentials in the large-$x$ limit}\label{largexap}

The fact that the conjugation matrix $G_\Phi$ of the RHP (\ref{RHP2})
has $1$ on the diagonal simplifies the matrix structure of the kernel
$I-G_\Phi$ of the integral equation formulation (\ref{inteq}) of this
problem. Combined with the explicit form (\ref{off}) of the
off-diagonal elements of $G_\Phi$, and Eqs.~(\ref{RHa}) and
(\ref{RHb}) for $\alf$ and $\bet$, the matrix integral equation gives
then the following equations for $\Phi_{11,+}(\lam),
\Phi_{12,+}(\lam)$:
\be\label{11} \Phi_{11,+}(\lam)=1-\frac{(1+e^{i\pi\kappa})}{2\pi
i}\inti\frac{\Phi_{12,+}(\mu)}{(\mu-\lam-i0)}\frac{\alf_-(\mu)}{
\bet_-(\mu)} \frac{ e^{-2i\mu x}}{(e^{\mu^2-\beta}-e^{i\pi\kappa}) }
\, ,\ \  \Imm \lam=0\, , \ee
and
\be\label{12} \Phi_{12,+}(\lam)=\frac{(1+e^{i\pi\kappa})}{2\pi
i}\inti\frac{ \Phi_{11,+}(\mu)}{(\mu-\lam-i0)}\frac{\bet_+(\mu)}{
\alf_+ (\mu)} \frac{e^{2i\mu x}}{(e^{\mu^2-\beta}-e^{i\pi\kappa})}\,
,\ \  \Imm \lam=0\, . \ee
An important feature of these equations is that the analytical
properties of $\alf (\lam)$ and $\bet (\lam)$ together with the fact
that $\Phi_{11}(\infty) =\alf(\infty)=\bet(\infty)=1$, make it
possible to obtain an estimate of $\Phi_{12,+}$ by shifting the
contour in the upper half-plane and evaluating the integral by the
sum of the residues. Strictly speaking, this approach requires the
functions $\Phi$ to conform to estimates of the type $|\Phi_{12}
(\lam,x,\beta)| \leq C/\lam$ and $|\Phi_{11}(\lam,\beta,x) -1| \leq
D/\lam$ for $\Imm \lam\geq 0$ in some ranges of $\beta$ and $x$:
$\beta<\beta_0\, ,x>x_0$, where $C,D$ depends only on the constants
$\beta_0\, , x_0$. In what follows, we assume that these estimates
hold, noting that they can be justified self-consistently as in Sec.
4 of \cite{IIK3}. The poles of the integrand in Eq.~(\ref{12}) are
given by $\lam+i0$ and the zeros $\lam_k$ of the function $e^{\lam^2
-\beta} -e^{i\pi\kappa}$ in the upper half-plane. The zeros of this
function are given by the formulae:
\be\label{roots1} (\Ree \lam_k)^2=\frac{1}{2} \left(\beta+ \sqrt{
\beta^2+\pi^2[\kappa+ 2k]^2}\right)\, ,\ \ (\Imm \lam_k)^2 =\frac{1
}{2}\left(-\beta+\sqrt{ \beta^2+ \pi^2[\kappa+ 2k]^2} \right) \, ,\ \
k=0,\pm 1,\pm 2,\cdots , \ee
and, as we will see below, the leading terms in the asymptotic
behavior are determined by $\lam_0^+$ and $\lam_{-1}^+$ which are
closest to the real axis:
\be\label{l0p}
\lam_0^+=\left(\beta+\sqrt{\beta^2+\pi^2\kappa^2}\right)^{1/2}/\sqrt{2}
+i\left(-\beta+\sqrt{\beta^2+\pi^2\kappa^2}\right)^{1/2}/\sqrt{2}\,
,\ \ \ee
and
\be \label{l1p}
\lam_{-1}^+=-\left(\beta+\sqrt{\beta^2+\pi^2[\kappa-2]^2}\right)^{1/2}
/\sqrt{2}+i\left(-\beta+\sqrt{\beta^2+\pi^2[\kappa-2]^2}\right)^{1/2}
/\sqrt{2}\, .\ \ \ee
The superscript $+$ denotes the zeros in the upper half-plane.
Closing the contour in this half-plane, we obtain
\be\label{int17}
\Phi_{12,+}(\lam)=(1+e^{i\pi\kappa})\Phi_{11,+}(\lam)\frac{\bet_+
(\lam )}{ \alf_+(\lam)}\frac{e^{2i\lam x}}{(e^{\lam^2-\beta}-
e^{i\pi\kappa})} +S^+(\lam)\, , \ee
with
\[ S^+(\lam)=\sum_{k=-\infty}^{+\infty}
\frac{(1+e^{i\pi\kappa})}{2e^{i\pi\kappa}}\frac{\bet(\lam_k^+)}{\alf
(\lam_k^+)} \frac{\Phi_{11}(\lam_k^+)e^{2i\lam_k^+ x}}{(\lam_k^+
-\lam) \lam_k^+} \, . \]
The series $S^+(\lam)$ is uniformly convergent for
$\lam\in\mathbb{R}$,  $x_0\leq x$, $0<\kappa\leq 1$, and with $\beta
\leq \beta_0$. Substituting Eq.~(\ref{int17}) into (\ref{11}), and
using the relations (\ref{RHa}) and (\ref{RHb}), we obtain the
following representation for Eq.~(\ref{11}):
\begin{eqnarray}\label{int18}
\Phi_{11,+}(\lam)&=&1-\frac{(1+e^{i\pi\kappa})^2}{2\pi i}\inti
\frac{\Phi_{11,+}(\mu)}{\mu-\lam-i0} \frac{d\mu}{ (e^{\mu^2 -\beta}
-e^{i\pi\kappa})(e^{\mu^2-\beta}+e^{i\pi\kappa}+2)}\nonumber\\
& &\ \ \ \ \ \ \ \ \ \ \ -\frac{1}{2\pi i}\inti \frac{R^+(\mu)}{
\mu-\lam-i0}d\mu\, , \end{eqnarray}
where
\[ R^+(\mu)=(1+e^{i\pi\kappa})\frac{\alf_-(\mu)}{\bet_-(\mu)}
\frac{e^{-2i\mu x}}{(e^{\mu^2-\beta}-e^{i\pi\kappa})}S^+(\mu)\, . \]
The form of the integral equation (\ref{int18}) suggests (see
Eq.~(\ref{aimp}) in Appendix \ref{srhp}) that it is useful to
consider the following inhomogeneous scalar RH problem for the
function $\tilde\Phi(\lam)$:
\be\label{RHP3}
\tilde\Phi_-(\lam)=\tilde\Phi_+(\lam)g(\lam)+R^+(\lam)\, ,\ \ \lam
\in\mathbb{R}\, ,\;\;\;\; \tilde\Phi(\infty)=1\, , \ee
such that $\tilde\Phi_+(\lam)=\Phi_{11,+}(\lam)$ for $\lam\in
\mathbb{R}$. The function $g$ in this RHP is given by the relation
\[ 1-g(\lam)=-\frac{(1+e^{i\pi\kappa})^2}{(e^{\lam^2-\beta}-
e^{i\pi\kappa})(e^{\lam^2-\beta}+e^{i\pi\kappa}+2)}\, , \]
which shows that using Eqs.~(\ref{RHa}) and (\ref{RHb}), $g(\lam)$
can be written as
\[ g(\lam)=\frac{(e^{\lam^2-\beta}+1)^2}{(e^{\lam^2-\beta}-
e^{i\pi\kappa})(e^{\lam^2-\beta}+e^{i\pi\kappa}+2)}= \frac{
\alf_-^{-1}(\lam)\bet_-^{-1}(\lam)}{\alf_+^{-1}(\lam)\bet_+^{-1}
(\lam)}\, . \]
This form of $g(\lam)$ combined with Eq.~(\ref{solution}) implies
that the solution of the RH problem (\ref{RHP3}) is:
\be\label{int19}
\tilde\Phi(\lam)=\alf^{-1}(\lam)\bet^{-1}(\lam)-\frac{\alf^{-1}
(\lam) \bet^{-1}(\lam)}{2\pi i} \inti\frac{\alf_-(\mu)\bet_-
(\mu)R(\mu)}{\mu-\lam}d\mu\, ,\ \ \ \lam\in\mathbb{C}/\mathbb{R}\, .
\ee
The functions $\tilde\Phi(\lam)$ and $\Phi_{11}(\lam)$ are analytic
in the upper half-plane. They have the same behavior at infinity:
$\tilde\Phi(\infty)=\Phi_{11}(\infty)=1$, and their boundary values
at the real axis $\tilde\Phi_+(\lam),\Phi_{11,+}(\lam)\,
,\lam\in\mathbb{R}$ are equal. This means that they coincide,
$\Phi_{11}(\lam)=\tilde\Phi(\lam)$, for $\Imm \lam>0$. Therefore,
from Eq.~(\ref{int19}) and the explicit expression for $R^+(\lam)$ we
have
\[ \Phi_{11}(\lam)=\alf^{-1}(\lam)\bet^{-1}(\lam)-\frac{\alf^{-1}
(\lam) \bet^{-1}(\lam)}{2\pi i}\sum_{k=-\infty}^{+\infty}
A(\lam_k^+,\lam)\Phi_{11}(\lam^+_k)\, ,\ \ \Imm \lam\geq 0\, , \]
where
\[ A(\lam_k^+,\lam)=\frac{(1+e^{i\pi\kappa})^2}{2\lam_k^+
e^{i\pi\kappa}}\frac{\bet(\lam_k^+)}{\alf(\lam_k^+)} e^{2i\lam^+_k
x}\inti\frac{\alf^2_-(\mu)e^{-2i\mu x}}{(e^{\mu^2-\beta} -e^{i
\pi\kappa})(\lam^+_k-\mu)(\mu-\lam)}d\mu\, . \]
For $x\rightarrow\infty$, the integral in the last equation can be
estimated as $ C|e^{-2i\lam_0^- x}|$, where $\lam_0^-$ is the $k=0$
zero (\ref{roots1}) in the lower half-plane, so the leading term of
$\Phi_{11}$ in the limit $x\rightarrow\infty$ is given by
\be\label{int20}
\Phi_{11}(\lam)=\alf^{-1}(\lam)\bet^{-1}(\lam)+O(e^{-4\Imm\lam_0^+
x})\, ,\ \ \Imm \lam>0\, , \ee
where we took into account that the magnitudes of the imaginary parts
of $\lam_0^{\pm}$ are the same. Using Eq.~(\ref{int20}) in
(\ref{12}), we obtain the leading term of $\Phi_{12}$ for large $x$:
\be\label{int21} \Phi_{12}(\lam)=\frac{(1+e^{i\pi\kappa})}{2\pi
i}\inti\frac{ \alf_+^{-2}(\mu)}{(\mu-\lam-i0)} \frac{e^{2i\mu
x}}{(e^{\mu^2-\beta}-e^{i\pi\kappa})}+O(e^{-4\Imm\lam_0^+ x})\, ,\ \
\Imm \lam>0\, . \ee
Substitution of the asymptotics (\ref{int20}) and (\ref{int21}) into
Eq.~(\ref{extract}) gives the large-distance asymptotic behavior of
the potentials $B_{++},B_{+-}$:
\be \label{pot1}
B_{+-}=\frac{1}{\pi}\inti\ln\left(\frac{e^{\mu^2-\beta}+1}{e^{\mu^2-
\beta}-e^{i\pi\kappa}}\right)d\mu+ O(e^{-4\Imm\lam_0^+ x})\, ,\;\;\;
x\rightarrow\infty\, , \ee
where the branch of the logarithm is fixed by the requirement that
$\ln(...)\rightarrow 0$ for $\mu \rightarrow \infty$, and
\[ B_{++}=\frac{(1+e^{i\pi\kappa})}{\pi}\inti\frac{\alf_+^{-2} (\mu)
e^{2i\mu x}}{ (e^{\mu^2-\beta}-e^{i\pi\kappa})}d\mu+O(e^{-4\Imm\lam_0^+
x})\, ,\;\;\;   x\rightarrow\infty\, . \]
The expression for $B_{++}$ can be made more precise by closing the
contour in the upper half-plane and evaluating the integral as the
sum of the residues
\be B_{++}=i(1+e^{-i\pi\kappa})\sum_{k=min}^{max}\frac{\alf^{-2}
(\lam_k^+) }{\lam_k^+}e^{2i\lam_k^+x} +O(e^{-4\Imm\lam_0^+ x})\, ,
\label{pot2} \ee
where the accuracy of our calculation of the leading term in the
asymptotics implies that the sum over poles can be limited to the
interval between $k=min(\beta,\kappa)<0$ and $k=max(\beta,\kappa)>0$
that are defined by the relations
\be\label{min} \Imm \lam_{min}^+<2\Im\lam_0^+\, ,\ \ \Imm
\lam_{min-1}^+>2\Imm \lam_0^+\, , \ee
and
\be\label{max} \Imm \lam_{max}^+<2\Im\lam_0^+\, ,\ \ \Imm
\lam_{max+1}^+>2 \Imm\lam_0^+\, . \ee
%

\subsection{Determination of $\sigma(x,\beta,\kappa)$ in the large-$x$
limit }

The large-distance asymptotics of the auxiliary potentials obtained
in the previous section and the differential equations (\ref{derxs})
and (\ref{derbs}) give the following asymptotic behavior of the
function $\sigma(x,\beta,\gamma) \equiv \ln \det(1-\gamma \hat K_T)$:
\be\label{sigmawc}
\sigma(x,\beta,\kappa)=-xC(\beta,\kappa)+c(\beta,\kappa)+
O(e^{-4\Imm\lam_0^+ x})\, ,\   x\rightarrow\infty\, , \ee
where  $C(\beta,\kappa)$ is defined by Eq.~(\ref{defcbk1}), and
$c(\beta,\kappa)$ is a constant that depends on  $\beta$ and
$\kappa$, but is independent of $x$. In order to determine this
constant, which through $\sigma(x,\beta,\kappa)$ determines the
$\beta$- and $\kappa$-dependence of the amplitudes of the asymptotic
terms in the expansion of the field correlators, we integrate the
relation
\be\label{degamma}
\partial_\gamma \sigma=-\inti R_T(\lam,\lam)\ d\lam\, ,
\ee
which is obtained by taking the derivative of Eq.~(\ref{repr}) with
respect to $\gamma$ and using the relation $(1-\gamma \hat
K_T)^{-1}\hat K_T=\hat R_T$. As shown in Sec.~\ref{sigma},
$R_T(\lam,\lam)$ in Eq.~(\ref{degamma})is:
\be\label{rll} R_T(\lam,\lam)=\frac{1}{2i}(\partial_\lam
f_+(\lam)f_-(\lam)- f_+(\lam)\partial_\lam f_-(\lam))\, . \ee
Our approach is similar to the one employed in Ref.~\cite{KKMST} (see also \cite{CZ}),
where the authors considered a generalized sine-kernel at zero
temperature and studied the large-distance asymptotic behavior of the
Fredholm determinant using the nonlinear steepest descent method of
Deift and Zhou \cite{DZ}. At zero temperature, $\det(1-\gamma \hat
K)$ becomes a particular case of the Fredholm determinant considered
in \cite{KKMST}.

The first step is the computation of $f_\pm(\lam)$  using the
relation (\ref{obtf}) and the following expression for the function
$\chi (\lam)$ in this relation:
\be\label{expchi} \chi(\lam)=\Phi(\lam)\left(\begin{array}{lr}
\bet(\lam)&0\\0&\alf( \lam)\end{array}\right)\, . \ee
In the previous Section, we have obtained Eqs.~(\ref{int20}) and
(\ref{int21}) for the  matrix elements $\Phi_{11}(\lam)$ and
$\Phi_{12}(\lam)$ of the matrix $\Phi(\lam)$ in the large-$x$ limit.
Closing the integration contour in the upper half plane in
Eq.~(\ref{int21}), one gets more explicitly:
\be\label{p12}
\Phi_{12}(\lam)=\pi\gamma\frac{\alf_+^{-2}(\lam)}{(e^{\lam^2-\beta}-
e^{i\pi\kappa})}e^{2i\lam x}+O(e^{-2\Imm\lam_0^+x})\, . \ee
One needs to find similar asymptotics for $\Phi_{21}$ and
$\Phi_{22}$. As in the previous Section, the integral formulation
(\ref{inteq}) of the RHP gives the following equations for these
matrix elements:
\be\label{ip21}
\Phi_{21,+}(\lam)=-\frac{\gamma}{2i}\inti\frac{\Phi_{22,+}(\mu) }{
\mu-\lam-i0}\frac{\alf_-(\mu)}{\bet_-(\mu)} \frac{e^{-2i\mu
x}}{(e^{\mu^2-\beta}-e^{i\pi\kappa})}d\mu\, ,\ \  \lam\in\mathbb{R}\,
, \ee
and
\be\label{ip22}
\Phi_{22,+}(\lam)=1+\frac{\gamma}{2i}\inti\frac{\Phi_{21,+}(\mu)}{
\mu-\lam-i0}\frac{\bet_+(\mu)}{\alf_+(\mu)} \frac{e^{2i\mu x}
}{(e^{\mu^2-\beta}-e^{i\pi\kappa})}d\mu\, ,\ \ \lam\in\mathbb{R}\, .
\ee
In the integral equation (\ref{ip21}), one can close the contour in
the lower half plane obtaining
\[ \Phi_{21,+}(\lam)=-\pi\gamma\sum_{k=-\infty}^\infty \frac{
\Phi_{22,+}(\lam_k^-)}{\lam_k^--\lam-i0}\frac{\alf_-(\lam_k^-)}{
\bet_-(\lam_k^-)} \frac{e^{-2|\Imm\lam_k^-|x}}{2\lam_k^-e^{i
\pi\kappa}}\, , \]
where $\lam_k^-$ are the zeroes of $e^{\mu^2-\beta}-e^{i\pi\kappa}$
in the lower half plane. Substituting this result in the integral
equation for $\Phi_{22,+}$ and closing the contour in the upper half
plane we find
\be\label{p22} \Phi_{22,+}(\lam)=1+O(e^{-2(|\Imm\lam_0^-|+\Imm
\lam_0^+)x})\, ,\ee
and, combining Eqs.~(\ref{p22}) and (\ref{ip21}),
\be\label{p21} \Phi_{21,+}(\lam)=O(e^{-2|\Imm\lam_0^-|x})\, . \ee
Now we are able to compute $f_\pm(\lam)$ up to exponentially small
corrections in $x$. Explicitly, Eqs.~(\ref{obtf}) and (\ref{expchi})
give:
\[ \left(\begin{array}{c} f_+(\lam)\\ f_-(\lam)\end{array} \right)=
\left(\begin{array}{c} \Phi_{11,+}(\lam)\bet_+(\lam)e_+(\lam)+
\Phi_{12,+}(\lam)\alf_+(\lam)e_-(\lam)\\ \Phi_{21,+} (\lam)
\beta_+(\lam)e_+(\lam)+\Phi_{22,+}(\lam) \alf_+(\lam)e_-(\lam)
\end{array}\right) , \]
and, therefore, using Eqs.~(\ref{int20}), (\ref{p12}), (\ref{p22}),
and (\ref{p21}) we find:
\be\label{expf}
\left(\begin{array}{c} f_+(\lam)\\ f_-(\lam)\end{array}\right)=
\left(\begin{array}{c}\alf_+^{-1}(\lam)e_+(\lam)e^{2\pi i\nuf(\lam)}
\\ \alf_+(\lam)e_-(\lam)\end{array}\right) . \ee
The function $\nuf(\lam)$ here (we have suppressed the dependence on
$\beta$) is defined in Eq.~(\ref{defnuf1}). One can see directly that
$\nuf(\lam)$ can be written in terms of the rescaled Fermi weight
$\theta(\lam)=(1+e^{\lam^2-\beta})^{-1}$ as
\[ \nuf(\lam)=-\frac{1}{2\pi i}\log(1-\pi\gamma\theta(\lam))=
\frac{1}{2\pi i}\log\left(\frac{e^{\lam^2-\beta}+1}{
e^{\lam^2-\beta}-e^{i\pi\kappa}}\right) , \]
and also that
\[ \alf_+(\lam)=\exp\left\{i\pi\nuf(\lam)+\mbox{P.V.} \inti
\frac{\nuf(\mu)}{\mu-\lam}d\mu\right\}\, . \]
Using Eqs.~(\ref{expf}) and (\ref{rll}) in Eq.~(\ref{degamma}), we
have
\begin{eqnarray}\label{intdg}
\partial_\gamma\sigma=-\inti R_T(\lam,\lam) d\lam &=& -\inti
\frac{1}{2\pi i}\frac{\pi\theta(\lam)}{(1-\pi\gamma \theta
(\lam))}(2ix -2\partial_\lam\log\alf_+(\lam)+2i\pi
\partial_\lam \nuf(\lam))\ d\lam\, ,\nonumber\\
&=&-\inti\partial_\gamma\nuf(\lam)(2ix
-2\partial_\lam\log\alf_+(\lam)+2i\pi \partial_\lam \nuf(\lam))\
d\lam\, .
\end{eqnarray}
Our current task is to write the RHS of Eq.~(\ref{intdg}) as a
derivative with respect to $\gamma$. The term depending on $x$
already has this form, but the last two, $x$-independent, terms do
not. We transform them as follows:
\begin{eqnarray*}
\inti\partial_\gamma\nuf(\lam)(2\partial_\lam\log\alf_+(\lam)- 2i\pi
\partial_\lam \nuf(\lam)) &=& \inti\partial_\gamma\nuf(\lam)
\partial_\lam \inti\left(\frac{\nuf(\mu)}{\mu-\lam-i0}+\frac{\nuf(\mu)
}{\mu-\lam+i0} \right)d\mu d\lam\, ,\nonumber\\ &=&\inti
\partial_\gamma \nuf(\lam)\left(\frac{\nuf(\mu)}{(\lam-\mu-i0)^2}
+\frac{\nuf(\mu)}{(\lam-\mu+i0)^2}\right)d\mu d\lam\, .
\end{eqnarray*}
Next, we show that this expression is equal to
\[ \partial_\gamma\inti\frac{\partial_\lam \nuf(\lam)\nuf(\mu)-
\nuf(\lam)\partial_\mu(\mu)}{2(\lam-\mu)}\ d\lam d\mu\, . \]
Indeed,
\begin{eqnarray*}
\partial_\gamma\inti \frac{\partial_\lam \nuf(\lam)\, \nuf(\mu)-
\nuf(\lam) \partial_\mu \nuf (\mu)}{2(\lam-\mu)}\ d\lam d\mu\
=\inti\frac{[(\partial_\gamma\partial_\lam\nuf(\lam)) \,
\nuf(\mu)+\partial_\lam \nuf (\lam)\partial_\gamma \nuf
(\mu)]}{\lam-\mu} d\lam d\mu\, ,\nonumber\\
=\frac{1}{2}\inti[(\partial_\gamma\partial_\lam\nuf(\lam))
\nuf(\mu)+\partial_\lam \nuf (\lam)\partial_\gamma \nuf (\mu)]
\left(\frac{1}{\lam-\mu+i0}+\frac{1}{\lam-\mu-i0}\right) d\lam d\mu\,
,\nonumber\\ =\frac{1}{2}\inti[\partial_\gamma\nuf(\lam)\nuf(\mu)
+\nuf(\lam)\partial_\gamma\nuf(\mu)] \left(\frac{1}{(\lam-\mu+i0)^2}
+\frac{1}{ (\lam-\mu-i0)^2}\right) d\lam d\mu\, ,\nonumber\\
=\inti\partial_\gamma\nuf(\lam)\left(\frac{\nuf(\mu)}{(\lam-
\mu-i0)^2}+\frac{\nuf(\mu)}{(\lam-\mu+i0)^2}\right)d\mu d\lam\, .
\end{eqnarray*}
Therefore,
\[ \partial_\gamma\sigma=-2ix\partial_\gamma\inti\nuf(\lam)d\lam
+ \partial_\gamma\inti\frac{\partial_\lam \nuf(\lam)\nuf(\mu)-
\nuf(\lam)\partial_\mu(\mu)}{2(\lam-\mu)}\ d\lam d\mu\, . \]
Integrating this relation with respect to $\gamma$ and taking into
account that $\sigma(\gamma=0)=0$, we obtain
\be\label{sigmapc} \sigma(x,\beta,\kappa)=-xC(\beta,\kappa)+
\inti\frac{\partial_\lam \nuf(\lam)\nuf(\mu)-\nuf(\lam)
\partial_\mu(\mu)}{2(\lam-\mu)}\ d\lam d\mu\, . \ee
Comparison with Eq.~(\ref{sigmawc}) shows finally that
\be\label{const1} c(\beta,\kappa)=\inti\frac{\partial_\lam \nuf(\lam)\nuf(\mu)
-\nuf(\lam)\partial_\mu \nuf(\mu)}{2(\lam-\mu)}\ d\lam d\mu\, . \ee

One can also obtain an alternative expression for the constant
$c(\beta,\kappa)$ by a simpler calculation. Using again the
differential equation (\ref{derbs}) and expressions for the
potentials (\ref{pot1}) and (\ref{pot2}), we have
\[ \sigma(x,\beta,\kappa)=-xC(\beta,\kappa)+\frac{1}{2}
\int_{-\infty}^\beta \left(\frac{dC(\beta',\kappa)}{d\beta'}
\right)^2 d\beta' +c(\kappa)+ O(e^{-4\Imm\lam_0^+ x})\, ,\
x\rightarrow\infty\, , \]
where $c(\kappa)$ is a constant that depends only on $\kappa$. The
initial condition $\sigma(x,-\infty,\kappa)=0$ (\ref{incs}) together
with the fact that the asymptotic behavior of $\sigma(x,\beta,
\kappa)$ is uniform in $\beta$, and $\kappa\in(0,1]$, imply that
$c(\kappa)=0$. This gives the following expression for
$c(\beta,\kappa)$:
\be\label{const2} c(\beta,\kappa)=\frac{1}{2}\int_{-\infty}^\beta
\left(\frac{dC(\beta',\kappa)}{d\beta'}\right)^2d\beta'\, . \ee
%

\subsection{Large-distance asymptotic behavior of the field-field
correlator}

As a reminder, we note that the anyonic field-field correlation
function can be written combining Eqs.~(\ref{fc}) and (\ref{dg}) as
follows:
\[ \la\fad(x_1)\fa(x_2)\ra_T=\frac{\sqrt{T}}{2\pi \gamma}B_{++}
(x,\beta,\kappa)e^{\sigma(x,\beta,\kappa)}\, , \ \ \gamma=(1+e^{i\pi
\kappa})/\pi\, . \]
Using in this equation the expressions obtained in the previous
Section for the potential $B_{++}$ and $\sigma$, and going back to
the original variables related to rescaled ones through
$x=x_{12}\sqrt{T}/2$ and $\beta=h/T$, we arrive at our main result
for the large-distance asymptotic behavior of the field-field
correlator
\be\label{asymptoticp}
\la\fad(x_1)\fa(x_2)\ra_T=e^{-x_{12}\frac{\sqrt{T}}{2}C(h/T,\kappa)}
e^{c(h/T,\kappa)} \sum_{k=min}^{max} c_k e^{ix_{12}\sqrt{T}\lam_k^+}
+O(e^{-2\Imm\lam_0^+ \sqrt{T}x_{12}})\, , \ee
where
\be \label{con2}
c_k=i\frac{e^{i\pi\kappa}\sqrt{T}}{2}\frac{\alf^{-2}(\lam_k^+)}{
\lam_k^+}\, . \ee
and $\lam_k^+$ are the zeros (\ref{roots1}) in the upper half-plane.
In the asymptotics (\ref{asymptoticp}), the functions
$C(\beta,\kappa)$, $c(\beta,\kappa)$, and $\alf(\lam)$, are defined,
respectively, by Eqs.~(\ref{defcbk1}), (\ref{con}) or (\ref{const2}), and
(\ref{defalf}), and the limits of summation over $k$ are given by
(\ref{min}) and (\ref{max}). The leading part of the asymptotics is
the term $k=0$, but as we approach the free fermionic point
$\kappa\rightarrow 1$, the $k=-1$ term also becomes relevant.

\section{Analysis of the results}\label{A}

As the last step, we check the validity of our main result in the
appropriate limits. When the statistics parameter $\kappa\rightarrow
0$ we should reproduce the results obtained for impenetrable bosons
in \cite{IIK1,IIK3,KBI}. Even though in the large-distance analysis
performed above we have not made any distinction between the cases of
negative and positive chemical potential (or, equivalently, $\beta$),
in what follows, we will see that in the bosonic limit, the
asymptotic behavior of the corrrelators is fundamentally different in
the two regions, as one would expect from \cite{IIK1,IIK3,KBI}. In
the fermionic limit $\kappa\rightarrow 1$, our result for the
field-field correlator of the anyons should reduce to the correlator
of the free fermions. At low temperatures, $\beta\rightarrow \infty$,
and positive chemical potential, the system becomes critical, and our
result should agree with the predictions of the conformal field
theory \cite{CM,PKA}.

\subsection{The bosonic limit}\label{bosoniclimit}

In the bosonic limit $\kappa\rightarrow 0$, Eq.~(\ref{defcbk1}) gives
in the case of positive and negative chemical potential,
respectively,
\be\label{int24} \lim_{\kappa\rightarrow 0} C(\beta,\kappa)=
\frac{1}{\pi}\inti\ln\left|\frac{e^{\lam^2-\beta}+1}{e^{\lam^2-
\beta}-1} \right|d\lam+2i\sqrt{\beta}\, ,\; \;\; \beta>0\, , \ee
and
\be\label{int25} \lim_{\kappa\rightarrow 0} C(\beta,\kappa)
=\frac{1}{\pi}\inti\ln\left(\frac{e^{\lam^2-\beta}+1}{e^{
\lam^2-\beta}-1}\right)d\lam\, ,\; \;\; \beta<0\, . \ee
Also, we have from Eq.~(\ref{l0p})
\be\label{int26} \lam_0^+=\sqrt{\beta}\,  ,\;\;  \beta>0\, ; \;\;\;\;
\lam_0^+=i\sqrt{|\beta|}\, ,\;\; \beta<0\,  . \ee
Introducing the function $C(\beta)$:
\be C(\beta)\equiv\frac{1}{\pi}\inti\ln\left|\frac{e^{\lam^2
-\beta}+1}{e^{\lam^2-\beta}-1}\right|d\lam\, ,\;\;\; \beta=h/T\, ,
\ee
and using Eq.~(\ref{asymptoticp}) reduced according to the previous
equations of this Section, we obtain the following asymptotic
behavior of the correlation function at negative chemical potential
\be\label{ii1} \la\fad(x_1)\fa(x_2)\ra_T\simeq \frac{T}{
2\sqrt{|h|}}\exp\left[a(h/T)+\frac{1}{2} \int_{-\infty }^{h/T}
\left(\frac{dC(\beta',\kappa)}{d\beta'}\right)^2d\beta'\right]
e^{-x_{12}\left[\frac{\sqrt{T}}{2} C(h/T)+\sqrt{|h|}\right]}\, , \ee
\[ h<0\, ,\ \ \ x_{12}\equiv (x_1-x_2)\rightarrow\infty\, , \]
where
\[ a(\beta)=\frac{|\beta|^{1/2}}{\pi}\inti\frac{d\mu}{\mu^2+
|\beta| }\ln\left(\frac{e^{\mu^2+|\beta|}-1}{e^{\mu^2+|\beta|}+1}
\right)\, , \]
and we have used (\ref{const2}) for $c(\beta,\kappa)$.
Similarly, for positive chemical potential we have
\be\label{ii2} \la\fad(x_1)\fa(x_2)\ra_T\simeq
e^{-x_{12}\frac{\sqrt{T}}{2}C(h/T)}\, ,\ \ \ h>0\, ,\ \ x_{12}\equiv
(x_1-x_2)\rightarrow\infty\, . \ee
Both asymptotics (\ref{ii1}) and (\ref{ii2}) agree with the result
for impenetrable bosons obtained in \cite{IIK2,IIK3,KBI}. In the
case of positive chemical potential we made this comparison for the
exponential terms only. When $\kappa=0$ and the chemical potential
$\beta$ is positive, the function $g_\alf(\lam)$ (\ref{RHa})
possesses an index, which means that we cannot transform the RHP as
in Section \ref{trans}. Also, as can be seen from Appendix
\ref{aoc}, $dC(\beta,\kappa)/d\beta$ which appears in (\ref{const2})
becomes divergent for $\kappa=0$ and $\beta\rightarrow 0$. This
means that the rigorous calculation of the pre-exponential factor
for positive chemical potential and $\kappa \rightarrow 0$ would
require a more sophisticated approach which we did not attempt in
this work.

\subsection{The fermionic limit}

In the model (\ref{hama}) with $c\rightarrow \infty$ we use to describe
the impenetrable anyons, in the limit of fermionic statistics
parameter $\kappa\rightarrow 1$, the anyonic field-field correlator
should coincide with that of free fermions. To see that
Eq.~(\ref{asymptoticp}) does indeed reduce to the free-fermionic
expression, we note that Eqs.~(\ref{defnuf1}) -- (\ref{con}) show
that for $\kappa= 1$,
\[ C(\beta,\kappa)=0\, , \;\;\; c(\beta,\kappa)=0\, .\]
This means that, in this case, $\alf(\lam)=1$, and the leading terms
in the correlator (\ref{asymptoticp}) take  the form:
\be\label{ff} \la\fad(x_1)\fa(x_2)\ra_T= \frac{-i\sqrt{T} }{2}
\sum_{k=-1,0} \frac{e^{ix_{12}\sqrt{T}\lam_k^+}}{ \lam_k^+} \, , \ee
Taking into account that for $\kappa= 1$, as follows from
Eqs.~(\ref{l0p}) and (\ref{l1p}), $\lambda_0^+ = -(\lambda_{-1}^+)^*
=(\beta+i\pi)^{1/2}$, one can see that Eq.~(\ref{ff}) coincides
exactly with the asymptotics of the field correlator for free
fermions, and reduces to
\be\label{fff} \la\fad(x_1)\fa(x_2)\ra_T= \frac{iT}{2} e^{-ax_{12}}
\sum_{\pm} \frac{\pm e^{\pm ib x_{12}}}{\sqrt{h\pm i\pi T}} \, , \ee
where
\be a=\left(\sqrt{h^2+\pi^2T^2}-h\right)^{1/2}/\sqrt{2}\, ,\;\;\;
b=\left(\sqrt{h^2+\pi^2T^2}+h\right)^{1/2}/\sqrt{2} \, . \ee
%

\subsection{The conformal limit}

For positive chemical potential and low temperatures, $\beta
\rightarrow \infty$, the system is conformal. The behavior of
$C(\beta,\kappa)$ in this limit is studied in Appendix \ref{aoc},
while the zeros (\ref{l0p}) and (\ref{l1p}) reduce to
\be\label{int28} \lam_0^+= \sqrt{\beta}+
i\frac{\pi\kappa}{2\sqrt{\beta}}\, , \;\;\;\; \lam_{-1}^+=
-\sqrt{\beta}+i\frac{\pi|\kappa-2|}{2\sqrt{\beta}}\,  . \ee
To make the connection with the conformal field-theory results
obtained in \cite{CM,PKA}, we use the original variables
$x=(x_1-x_2)\sqrt{T}/2, \beta=h/T$ and the fact that in the notations
used in this work, the Fermi momentum  and velocity are,
respectively, $k_F=\sqrt{h}$ and $v_F=2\sqrt{h}$. Then, using
Eqs.~(\ref{asymptoticp}) and (\ref{int28}), we obtain the following
result for the asymptotics for the field-field correlator at low
temperatures. The accuracy of our calculation of the asymptotics
(\ref{asymptoticp}), indicated by the last term in this equation, and
leading to the conditions (\ref{min}) and (\ref{max}), implies that
for $0<\kappa<2/3$ we can keep only one term in the asymptotic
expansion of the correlator:
\be\label{smallk} \la\fad(x_1)\fa(x_2)\ra_T\simeq c_0 e^{-x_{12}
\frac{\pi T}{v_F} \left(\frac{\kappa^2}{2} +\frac{1}{2}\right)}
e^{ix_{12}k_F\kappa}\, .  \ee
For $2/3<\kappa<1$, however, the two terms in the expansion are
legitimate, giving
\be \label{largek} \la\fad(x_1)\fa(x_2)\ra_T\simeq c_0
e^{-x_{12}\frac{\pi T}{v_F} \left(\frac{\kappa^2}{2}+ \frac{1}{2}
\right)} e^{ix_{12}k_F\kappa} +c_{-1}e^{-x_{12} \frac{\pi T}{v_F}
\left[2\left( \frac{ \kappa }{2}-1 \right)^2 +\frac{1}{2}\right]}
e^{ix_{12}k_F(\kappa-2)}\, . \ee
The conformal result obtained in (\cite{PKA,CM}) is
\begin{eqnarray}\label{conformal}
\la\fad(x_1)\fa(x_2)\ra_T\simeq\sum_{Q=\{N^\pm,d\}}B(Q)
e^{-x_{12}\frac{\pi T}{v_F}\left[2N^++2N^-+\frac{1}{2}+2\left(
d+\frac{\kappa}{2}\right)^2\right]} e^{ix_{12}k_F \left(
2d+\kappa\right)}\, .
\end{eqnarray}
One can see that the leading terms in the sum (\ref{conformal}),
which correspond to $Q={0,0,0}$ and $Q={0,0,-1}$ , are identical
(modulo some constants) to (\ref{smallk}) and (\ref{largek}). The
presence of the second term in Eq.~(\ref{largek}) (and the term with
$Q={0,0,-1}$ in the conformal expansion) explains why close to the
fermionic point the correlation function exhibits beatings, see also
Eq.~(\ref{fff}). At $T=0$, this feature of the correlation function
was noticed and explained by Calabrese and Mintchev \cite{CM}.

\section{Conclusions}

We have computed rigorously the large-distance asymptotic behavior
of the field-field correlation functions of impenetrable anyons at
finite temperature. In the process, we have also obtained a system
of differential equations which characterize the correlators at any
distance. Our result agrees with the predictions of the conformal
field theory at low temperatures, and describes a transition in the
behavior of the correlators between the impenetrable bosons realized
for the statistical parameter $\kappa=0$, and free fermions at
$\kappa=1$. Variation of the asymptotic behavior of the field-field
correlation function with the statistics parameter $\kappa$
illustrates the role of the exchange statistics in one dimension.
This variation contradicts the intuitive notion that exchange
statistics is irrelevant in 1D systems of impenetrable particles,
because they can not be exchanged in 1D geometry.

The next important step in the analysis of the 1D anyons would be to
extend the results of this work to the description of the time
dependence of the field-field correlators. The resulting time-,
space-, and temperature-dependent correlation functions could be
calculated along the lines of Ref.~\cite{IIKV}, which treated the
case of impenetrable bosons. We will attempt to do this in a future
publication.

\appendix

\section{Short-Distance Asymptotics}\label{sda}

In this appendix, we obtain the short-distance asymptotics of the
potentials $B_{+-}$ and $B_{++}$ and as a byproduct, compute the
short-distance asymptotics of the field correlator. First, we need
to express the potentials in the form more amenable for
short-distance computations. We start with $B_{+-}$. Multiplying
both sides of the integral equation (\ref{df}) which defines
$f_+(\lam)$ by $\gamma e_-(\lam)$ one gets
\begin{eqnarray*}
\gamma f_+(\lam)\sqrt{\vartheta(\lam)}e^{-i\lam x}&=& \gamma
\vartheta( \lam)+\gamma^2\sqrt{\vartheta(\lam)}e^{-i\lam x}
\inti\sqrt{\vartheta(\lam)} \frac{\sin x(\lam-\mu)}{(\lam- \mu)}
\sqrt{\vartheta(\mu)} f_+(\mu)\ d\mu\, ,\nonumber\\&=& \gamma
\vartheta(\lam)+\gamma\vartheta(\lam)\inti\frac{1-e^{-2i(\lam-\mu)
x}}{2i(\lam-\mu)}\gamma f_+(\mu) \sqrt{\vartheta(\mu)} e^{-i\mu
x}d\mu\, . \end{eqnarray*}
This equation shows that $B_{+-}$ can be written as
\be\label{bpm} B_{+-}(x,\beta,\gamma)=\inti s(\lam) d\lam \, , \ee
where the function $s(\lam)$ solves the following integral equation
\be\label{ds}
s(\lam)-\gamma\vartheta(\lam)\inti\frac{1-e^{-2i(\lam-\mu)x}}{2i(\lam
-\mu)}s(\mu) d\mu=\gamma\vartheta(\lam)\, . \ee
In a similar fashion, we get
\be\label{bpp}
B_{++}(x,\beta,\gamma)=\inti e^{2i\lam x} s(\lam) d\lam\, ,
\ee
where $s(\lam)$ is the solution of the same integral equation (\ref{ds}).

For small $x$, the solution of (\ref{ds}) can be expanded as
\[ s(\lam)\equiv s(\lam,\beta,\gamma)=\sum_{k=0}^\infty s_k(\lam,
\beta, \gamma) x^k\, , \]
where $s_k$ are defined by the following recursion relations
\be \label{recurrence} s_0(\lam)=\gamma\vartheta(\lam)\, ,\;\;\;\;
s_m(\lam) = s_0(\lam) \sum_{k=0}^{m-1}\frac{(2i)^{m-k-1}}{(m-k)!}
\inti(\mu- \lam)^{m-k-1} s_k(\mu)d\mu\,,\ \ m\ge 1\, . \ee
Defining
\[ \beta_l(\beta,\gamma)=\gamma\inti\lam^l\vartheta(\lam)d\lam\, , \]
i.e., $\beta_l \equiv 0$ for odd $l$, and using (\ref{recurrence})
in (\ref{bpm}) and (\ref{bpp}), we obtain the short-distance
asymptotics for the potentials
\begin{eqnarray*}
B_{++}(x,\beta,\gamma)&=&\beta_0+\beta_0^2x+\left(\beta_0^3-2
\beta_2 \right)x^2 +\left(\beta_0^4-\frac{4}{3}\beta_0 \beta_2
\right)x^3+O(x^4)\, ,\\ B_{+-}(x,\beta,\gamma)&=& \beta_0+\beta_0^2x
+\beta_0^3x^2+\left(\beta_0^4-\frac{4}{3}\beta_0\beta_2\right)x^3
+O(x^4)\, . \end{eqnarray*}
These relations give us the short-distance asymptotics of the field
correlator. Combining them with Eqs.~(\ref{derbs}) and
(\ref{fullexp})
\[ g(x,\beta,\gamma)=B_{++}(x,\beta,\gamma)e^{\sigma(x,\beta,
\gamma)}|_{\gamma=(1+e^{i\pi\kappa})/\pi}\, ,\ \
\sigma(x,\beta,\gamma)=-\int_0^x B_{+-}(y,\beta,\gamma) dy\, ,\]
we find first
\[ \sigma(x,\beta,\gamma)=-\beta_0 x-\frac{1}{2}\beta_0^2 x^2
-\frac{1}{3}\beta_0^3 x^3+O(x^4)\, , \]
and then
\[ g(x,\beta,\gamma)=\beta_0\left(1-2\frac{\beta_2}{\beta_0}x^2
+\frac{2}{3}\beta_2 x^3\right)+O(x^4)\, . \]
In the original variables $x=(x_1-x_2)\sqrt{T}/2>0,\ \beta=h/T, \
\lam\rightarrow \lam/\sqrt{T}$, this result gives for the correlator
(\ref{fc})
\be\label{sdasympt}
\la\fad(x_1)\fa(x_2)\ra_T=D\left(1-\frac{E}{2D}(x_1-x_2)^2+\gamma
\frac{\pi E}{6}(x_1-x_2)^3\right)+O\left((x_1-x_2)^4\right)\, , \ \
\ \gamma=(1+e^{i\pi\kappa})/\pi \, , \ee
where
\[ D=\frac{1}{2\pi}\inti\frac{d\lam}{1+e^{(\lam^2-h)/T}}\, ,\ \ \
E=\frac{1}{2\pi}\inti\frac{\lam^2d\lam}{1+e^{(\lam^2-h)/T}}\, , \]
are the particle and the kinetic energy density, respectively.

\section{Low Density Expansions}\label{lde}

As usual, the low-density limit is reached when the chemical
potential is such that $\beta=h/T \rightarrow-\infty$. In our
rescaled variables, the density of impenetrable anyons is given by
\[ D=\frac{\sqrt{T}}{2\pi}\inti\frac{d\lam}{1+e^{\lam^2-\beta}}\, ,  \]
so that $D\rightarrow 0$ for $\beta\rightarrow -\infty$. In what
follows, it will be convenient to use the variable
\[ \zeta \equiv -e^\beta,\ \ \ \ \zeta\rightarrow 0\ \ \ \mbox{for}
\ \ \beta\rightarrow -\infty\, .
\]
In order to obtain the low-density expansions for the potentials
\[ B_{++}=\sum_{k=1}^\infty b_k(x)\zeta^k\, , \ \ \ B_{+-}=
\sum_{k=1}^\infty c_k(x)\zeta^k\, , \]
we again use Eqs.~(\ref{bpm}) and (\ref{bpp}), and the integral
equation (\ref{ds}) which has the form suitable for interation
expansion in density. In terms of $\zeta$, the Fermi weight can be
represented as
\[ \vartheta(\lam)=-\sum_{k=1}^\infty\zeta^ke^{-k\lam^2} . \]
Expanding also $s(\lam)$:
\[ s(\lam)=\sum_{k=1}^\infty \zeta^k s_k(\lam,x)\, , \]
we obtain from Eq.~(\ref{ds}) the following recursion relations for
the ``coefficients'' $s_k$:
\begin{eqnarray*}
s_1(\lam)&=&-\gamma e^{-\lam^2} , \\ s_k(\lam,x) &=& e^{-
\lam^2}s_{k-1}(\lam,x)-\gamma e^{-\lam^2} \inti \frac{1- e^{-2i(
\lam-\mu)x}}{2i(\lam-\mu)}s_{k-1}(\mu,x) d\mu\, ,\ \ k\ge 2\, .
\end{eqnarray*}
The first terms of the expansions of the potentials obtained from
these recursion relations and Eqs.~(\ref{bpm}) and (\ref{bpp}) are:
\begin{eqnarray}
B_{+-}(x,\zeta,\kappa)&=&-\gamma\sqrt{\pi}\zeta+\left(-\gamma \sqrt{
\frac{\pi}{2}}+\gamma^2\pi\int_0^x e^{-x_1^2}dx_1\right)\zeta^2+
O(\zeta^3)\, ,\nonumber \\ B_{++}(x,\zeta,\kappa)&=&-\gamma
\sqrt{\pi} e^{-x^2} \zeta+ \left(-\gamma\sqrt{\frac{\pi}{2}}
e^{-x^2} +\gamma^2 \pi e^{-x^2}\int_0^x e^{-2x_1^2+2x_1x}dx_1\right)
\zeta^2+ O(\zeta^3)\, . \label{int10} \end{eqnarray}
Similarly to the short-distance expansions, Eq.~(\ref{int10}) gives
\[ \sigma(x,\beta,\gamma)=-\gamma\sqrt{\pi}xe^\beta+O(e^{2\beta})\,
, \]
and the correlator in the rescaled variables:
\[ \frac{\sqrt{T}}{2\pi\gamma}g(x,\beta,\gamma)=\frac{\sqrt{T}}{2
\pi^{1/2}}e^{-x^2}e^\beta+O(e^{2\beta})\, . \]
In the original variables, this result for the correlator is:
\be \la\fad(x_1)\fa(x_2)\ra_T=De^{-T(x_1-x_2)^2/4} . \ee
It is valid as long as we can neglect the $O(e^{2\beta})$ terms,
i.e. for $T(x_1-x_2)^2\ll |h|/T$.

\section{Solvability of the Matrix Riemann-Hilbert Problem}
\label{smrhp}

As we have shown in Section \ref{connect}, the matrix RH problem
(\ref{RHP}) is equivalent to the system of nonsingular integral
equations for functions $f_\pm(\lam)$:
\[ f_\pm(\lam)-\gamma\int_{-\infty}^{+\infty}K_T(\lam,\mu)f_\pm
(\mu)d\mu =e_\pm(\lam)\, , \]
with the kernel (\ref{ad1}). This means that the RH problem has a
unique solution whenever this system of integral equations of
Fredholm type has a unique solution. To analyze these equations, we
fix $\beta$ and $\kappa$, leaving the kernel $K_T(\lam,\mu)$ a
function of coordinate $x$. Let $D$ be an open connected subset of
the complex plane and $\cal{L}(\cal{H})$ -- the space of operators
acting on a separable Hilbert space $\cal{H}$. Consider the function
\[ f(x):D\rightarrow\cal{L}(\cal{H})\, , \]
which for each $x$ in $D$ gives the integral operator with kernel
$K_T(\lam,\mu)$. Then for each $x$ in the finite strip
\[ 0<a<\Ree x<b\, , \ \ \Imm x<\epsilon\, , \]
$f(x)$ is analytic operator-valued function. The kernel
$K_T(\lam,\mu)$ also satisfies the estimate
\[ \inti\inti |K_T(\lam,\mu)|^2d\lam\ d\mu<Cb^2 , \]
where $C$ is a constant, which means that $f(x)$ is compact for each
$x\in D$ (see Thm. VI. 23 of \cite{RS}). Under these conditions, we
can apply the analytic Fredholm theorem.
\begin{thm}[Thm VI. 14 of \cite{RS}]
Let $D$ be an open connected subset of $\mathbb{C}$. Let
$f:D\rightarrow\cal{L}(\cal{H})$ be an analytic operator-valued
function such that for each $z\in D$, $f(z)$ is compact. Then,
either
\begin{itemize}
\item[a)] $(I-f(z))^{-1}$ exists for no $z\in D$; or
\item[b)] $(I-f(z))^{-1}$ exists for all $z\in D\backslash S$,
where $S$ is  a discrete subset of $D$ (i.e. a set which has no
limit points in $D$). In this case $(I-f(z))^{-1}$ is meromorphic in
$D$, analytic in $D\backslash S$, the residues at the poles are
finite rank operators, and if $z\in S$, then equation
$f(z)\psi=\psi$ has a nonzero solution in $\cal{H}$.
\end{itemize}
\end{thm}
As a consequence of this theorem, we have to prove that for at least
one point in the strip $D$ the integral equations have a unique
solution. But this is definitely true for small $x$, where the
Liouville-Neumann series is convergent. Thus, we have shown that the
matrix RH problem has a unique solution except for a countable set
of values of $x_n$, which we will denote by $X=\{x_n\}.$

\section{Scalar Riemann-Hilbert Problem}\label{srhp}

This Appendix provides the basic information on the scalar
Riemann-Hilbert problem used in Sec.~\ref{trans}. The general problem of
this type for the semi-plane is formulated as follows \cite{G}. Consider two
functions $g(\lam)$ and $r(\lam)$  defined on the real axis,
with $g(\lam)$ nonvanishing. Both are assumed to satisfy the H\"older
condition: $|g(\lam_1)-g(\lam_2)|<C |\lam_1-\lam_2|^k$, and similarly
for $r(\lam)$, with some power $k$: $0<k\le 1$. One needs to find the
function $\alf(\lam)$, or $\tilde\alf(\lam)$, which is analytic
separately in the upper and lower half-plane, with the boundary
values on the real axis satisfying the conditions:
\be\label{hom} \alf_-(\lam)=\alf_+(\lam)g(\lam)\, ,\ \
\lam\in\mathbb{R}\ \ \mbox{ homogeneous problem, } \ee or
\be\label{inhom}
\tilde\alf_-(\lam)=\tilde\alf_+(\lam)g(\lam)+r(\lam)\, ,\ \
\lam\in\mathbb{R}\ \ \mbox{ inhomogeneous problem. } \ee
For the purposes of this work, we will assume also the normalization
condition $\alf(\infty)=\tilde\alf(\infty)=1$. The considerations
presented below can be extended to the more general case of a
simply-connected closed contour in the complex plane -- see
\cite{G}.

\subsection{The homogeneous problem}

We need to distinguish three cases depending on the index
$\chi(g)=(1/2\pi) \mbox{Var}_{[-\infty,+\infty]} \mbox{arg}\
g(\lam)$ of the function $g(\lam)$. If $\chi=0$, the RH problem with
the normalization condition is uniquely solvable. If the index is
positive, $\chi>0$, the problem has $\chi+1$ linearly independent
solutions, whereas the problem has no solution for $\chi<0$. For
$\chi=0$, which is the situation most important for the present
discussion, the solution of the RH problem (\ref{hom}) is given by
\be \alf(\lam)=\exp\left\{-\frac{1}{2\pi i}\inti\frac{\ln
g(\mu)}{\mu- \lam}d\mu\right\}\, ,\ \ \lam\in\mathbb{C}/
\mathbb{R}\, . \ee
As in the matrix case, it is straightforward to show that the scalar
RH problem (\ref{hom}) is equivalent to the singular integral
equation
\[ \alf_+(\lam)=1+\frac{1}{2\pi i}\inti\frac{\alf_+(\mu)(1-g(\mu) )
}{ \mu-\lam-i0}d\mu\, ,\ \ \lam\in\mathbb{R}\, .\]

\subsection{The inhomogeneous problem}

Similarly to the homogeneous problem, the solution of the inhomogeneous
RH problem (\ref{inhom}) with the normalization condition is unique
for $\chi=0$, which is the situation of interest for the present
discussion. The solution can be obtained from the solution of the
homogeneous problem with the same $g(\lam)$. If $\alf(\lam)$ solves
(\ref{hom}), then  $g(\lam)=\alf_-(\lam)/\alf_+(\lam)$ for
$\lam\in\mathbb{R}$, and (\ref{inhom}) can be written as
\[ \frac{\tilde\alf_+(\lam)}{\alf_+(\lam)}-\frac{\tilde\alf_- (\lam)
}{\alf_-(\lam)}=-\frac{r(\lam)}{\alf_-(\lam)}\, . \]
The functions $\tilde\alf_+(\lam)/\alf_+(\lam)$ and
$\tilde\alf_-(\lam)/\alf_-(\lam)$ are the boundary values of the
function $\tilde\alf(\lam)/\alf(\lam)$ which is analytic in the
complex plane minus the real axis and approaches $1$ at infinity due
to imposed normalization conditions $\tilde\alf(\infty)=
\alf(\infty)=1$. Using the properties of the Cauchy integral, one
obtains from this:
\[ \frac{\tilde\alf(\lam)}{\alf(\lam)}=1-\frac{1}{2\pi i}\inti
\frac{r(\mu)}{\alf_-(\mu)(\mu -\lam)}d\mu\, , \ \ \
\lam\in\mathbb{C} /\mathbb{R}\, . \]
Thus, solution of the inhomogeneous scalar RH problem (\ref{inhom})
is given by
\be\label{solution} \tilde\alf(\lam)=\alf(\lam)\left(1-\frac{1}{2\pi
i}\inti\frac{r(\mu) }{\alf_-(\mu)(\mu-\lam)}d\mu\right) , \ \ \
\lam\in\mathbb{C}/\mathbb{R}\, . \ee
where $\alf(\lam)$ is the solution of the homogeneous problem
(\ref{hom}). The singular integral equation equivalent to the
inhomogeneous RH problem is
\be\label{aimp} \tilde\alf_+(\lam)=1+\frac{1}{2\pi
i}\inti\frac{\tilde\alf_+(\mu)(1-g(\mu))}{\mu-\lam-i0}d\mu
-\frac{1}{2\pi i}\inti\frac{r(\mu)}{\mu-\lam-i0}d\mu \, ,\ \
\lam\in\mathbb{R}\,. \ee
%

\section{Analysis of $C(\beta,\kappa)$}\label{aoc}

In this appendix, we study the behavior of the function
$C(\beta,\kappa)$, which enters the asymptotics (\ref{asymptoticp})
of the field correlator, and is defined by Eq.~(\ref{defcbk1}) with
condition (\ref{log}), for large and small $\beta$. We start with
$\beta<0$. In this case,  the expansion of the logarithms in
Eq.~(\ref{defcbk1}) according to $\ln(1-z)=-\sum_{n=1}^\infty z^n/n$
for $|z|<1$, gives:
\begin{eqnarray*}
C(\beta,\kappa)&=&\frac{1}{\pi}\inti\sum_{n=1}^\infty\left(\frac{
(-1)^{n+1} e^{-n(\lam^2+|\beta|)}}{n}+ \frac{e^{i n\pi \kappa}e^{-n
(\lam^2+|\beta|)}}{n}\right)d\lam\, ,\nonumber\\ &=&\frac{1}{\sqrt{
\pi} }\sum_{n=1}^\infty\left(\frac{(-1)^{n+1}e^{-n|\beta|}}{n^{3/2}}+
\frac{e^{i n\pi\kappa}e^{-n|\beta|}}{n^{3/2}}\right)\, .
\end{eqnarray*}
Therefore, the leading terms for large and small $|\beta|$ are,
respectively,
\be
 C(\beta,\kappa)=\frac{e^{-|\beta|}}{\sqrt{\pi}}(1+\cos \pi\kappa)
+i\frac{e^{-|\beta|}}{\sqrt{\pi}}\sin \pi\kappa\, ,\ \ \beta
\rightarrow -\infty \, , \ee
and
\be C(\beta,\kappa)=\frac{1}{\sqrt{\pi}}\sum_{n=1}^\infty\left(\frac{
(-1)^{n+1}+\cos n\pi\kappa}{n^{3/2}}\right) +i\frac{1}{\sqrt{\pi}}
\sum_{n=1}^\infty\frac{\sin n\pi\kappa}{n^{3/2}}\, ,\ \ \beta
\rightarrow 0\, . \label{small} \ee
When $\beta$ is small, its sign is irrelevant and, as we see
explicitly below, the last equation holds both for negative and
positive $\beta$.

For $\beta>0$, we can transform the logarithms in Eq.~(\ref{defcbk1})
so that the same expansion is applicable, and get
\begin{eqnarray}
C(\beta,\kappa)&=&\frac{2}{\pi}\int_0^{\sqrt{\beta}}\left[ -i\pi(
\kappa-1)+\sum_{n=1}^\infty\left(\frac{(-1)^{n+1}e^{n(\lam^2-\beta
)}}{n}+ \frac{e^{-i n\pi\kappa}e^{n(\lam^2-\beta)}}{n}\right) \right]
d\lam \nonumber\\ & &\ \ \ \ +\frac{2}{\pi}\int_{\sqrt{\beta}}^\infty
\sum_{n=1}^\infty \left(\frac{(-1)^{n+1}e^{-n(\lam^2-\beta)}}{n}+
\frac{e^{i n\pi\kappa}e^{-n(\lam^2-\beta)}}{n}\right)d\lam\, .
\label{int23} \end{eqnarray}
When $\beta$ is large, the formulae
\[ e^{-\beta n}\int_0^{\sqrt{\beta}}e^{\lam^2 n}d\lam=\frac{1}{2n
\sqrt{\beta}}+O\left(\frac{1}{\beta^{3/2}}\right) ,\ \ \ e^{\beta
n}\int_{\sqrt{\beta}}^\infty e^{-\lam^2 n}d\lam=\frac{1}{
2n\sqrt{\beta}}+O\left(\frac{1}{\beta^{3/2}}\right) , \]
simplify this expression into
\[ C(\beta,\kappa)=\frac{2}{\pi\sqrt{\beta}}\sum_{n=1}^\infty\left(
\frac{(-1)^{n+1}}{n^2}+\frac{\cos n\pi\kappa}{n^2}\right)
-2i\sqrt{\beta}(\kappa-1)+O\left(\frac{1}{\beta^{3/2}}\right)\, . \]
This expression can be transformed finally using the formulae (0.234)
and (1.443) of \cite{GR}, $\sum_{k=1}^\infty(-1)^{n+1}/n^2=\pi^2/12$,
and $\sum_{k=1}^\infty\cos n\pi\kappa/n^2=\pi^2 B_2(\kappa/2)$, where
$B_2(x)=x^2-x+1/6$ is the second Bernoulli polynomial. This gives
\be C(\beta,\kappa)=\frac{\pi}{\sqrt{\beta}}\left(\frac{\kappa^2}{2}-
\kappa+\frac{1}{2}\right) -2i\sqrt{\beta}(\kappa-1)\, ,\ \ \
(\beta\rightarrow\infty)\, . \ee
Equation (\ref{int23}) also shows that for $\beta \rightarrow +0$,
$C(\beta,\kappa)$ is given by the same Eq.~(\ref{small}) as for
$\beta \rightarrow -0$.


\end{document}